\documentclass[12pt]{article}

\usepackage{booktabs}

\usepackage[T1, OT1]{fontenc}
\DeclareTextSymbolDefault{\DH}{T1}
\usepackage{amssymb,latexsym}
\usepackage{graphicx,amsmath,amsfonts}
\usepackage{comment}
\usepackage{tikz}
\usetikzlibrary{arrows}
\usetikzlibrary{calc}
\usepackage{pgfplots}
\usepackage{chngcntr}
\usepackage{apptools}
\AtAppendix{\counterwithin{lemma}{section}}
\AtAppendix{\counterwithin{theorem}{section}}
\AtAppendix{\counterwithin{proposition}{section}}
\AtAppendix{\counterwithin{example}{section}}

\usepackage{url}
\usepackage{mathrsfs}
\usepackage{colortbl}
\definecolor{LightCyan}{rgb}{0.88,1,1}
\definecolor{celadon}{rgb}{0.67, 0.88, 0.69}
\definecolor{columbiablue}{rgb}{0.61, 0.87, 1.0}
\usepackage{paralist}
\usepackage{hyperref}

\usepackage[sort&compress,nameinlink]{cleveref}
\usepackage{nicefrac}
\usepackage{booktabs}
\usepackage{mathtools}
\usepackage{setspace}
\usepackage{etoolbox}
\usepackage{wrapfig}
\usepackage[square,numbers,sort&compress]{natbib}
\usepackage{amsthm}
\usepackage{makecell}
\allowdisplaybreaks
\usepackage[nottoc,numbib]{tocbibind}

\makeatletter

\patchcmd{\NAT@citex}
  {\@citea\NAT@hyper@{%
     \NAT@nmfmt{\NAT@nm}%
     \hyper@natlinkbreak{\NAT@aysep\NAT@spacechar}{\@citeb\@extra@b@citeb}%
     \NAT@date}}
  {\@citea\NAT@nmfmt{\NAT@nm}%
   \NAT@aysep\NAT@spacechar\NAT@hyper@{\NAT@date}}{}{}

\patchcmd{\NAT@citex}
  {\@citea\NAT@hyper@{%
     \NAT@nmfmt{\NAT@nm}%
     \hyper@natlinkbreak{\NAT@spacechar\NAT@@open\if*#1*\else#1\NAT@spacechar\fi}%
       {\@citeb\@extra@b@citeb}%
     \NAT@date}}
  {\@citea\NAT@nmfmt{\NAT@nm}%
   \NAT@spacechar\NAT@@open\if*#1*\else#1\NAT@spacechar\fi\NAT@hyper@{\NAT@date}}
  {}{}

\makeatother

\usepackage{pifont}

\newtheorem{theorem}{Theorem}

\newtheorem{corollary}{Corollary}

\newtheorem{definition}{Definition}
\newtheorem{example}{Example}
\newtheorem{proposition}{Proposition}

\usepackage{etoolbox}
\newcommand{\addQEDstyle}[2]{\AtBeginEnvironment{#1}{\pushQED{\qed}\renewcommand{\qedsymbol}{#2}}
\AtEndEnvironment{#1}{\popQED}}
\addQEDstyle{example}{$\dashv$}

\newcommand{\reals}{\mathbb R}

\newcommand{\naturals}{\mathbb N}

\newcommand{\rationals}{\mathbb Q}

\newcommand{\calR}{\mathcal{R}}

\newcommand{\pav}{{{\mathrm{PAV}}}}
\newcommand{\geom}{{{\mathrm{geom}}}}
\newcommand{\np}{{{\mathrm{NP}}}}

\newcommand{\sat}{{{{\mathrm{sat}}}}}

\newcommand{\shortcite}[1]{(\citeyear{#1})}

\makeatletter
\newtheorem*{rep@theorem}{\rep@title}
\newcommand{\newreptheorem}[2]{%
\newenvironment{rep#1}[1]{%
 \def\rep@title{#2 \ref{##1}}%
 \begin{rep@theorem}}%
 {\end{rep@theorem}}}
\makeatother
\newreptheorem{theorem}{Theorem}
\newreptheorem{proposition}{Proposition}
\newreptheorem{lemma}{Lemma}
\newreptheorem{corollary}{Corollary}

\newcommand{\lambdascore}{{{{\lambda\text{-}\mathrm{score}}}}}

\definecolor{OKgreen}{HTML}{035925}
\definecolor{NOred}{HTML}{8F061F}

\newcommand{\notstarted}

\usepackage[leftbars]{changebar}
\crefname{abcrule}{Rule}{Rules}

\begin{document}

\title{Phragm\'{e}n Rules for Degressive and Regressive Proportionality}

\author{Michał Jaworski\\
  University of Warsaw\\
  Warsaw, Poland\\
  {\small \texttt{m.jaworski@mimuw.edu.pl}}
  \and 
Piotr Skowron\\
  University of Warsaw\\
  Warsaw, Poland\\
  {\small \texttt{p.skowron@mimuw.edu.pl}}
}
\date{}

\maketitle

\begin{abstract}
We study two concepts of proportionality in the model of approval-based committee elections. In degressive proportionality small minorities of voters are favored in comparison with the standard linear proportionality. Regressive proportionality, on the other hand, requires that larger subdivisions of voters are privileged. We introduce a new family of rules that broadly generalize Phragmén's Sequential Rule spanning the spectrum between degressive and regressive proportionality. We analyze and compare the two principles of proportionality assuming the voters and the candidates can be represented as points in an Euclidean issue space.
\end{abstract}

\setcounter{tocdepth}{2}

\section{Introduction}\label{sec:intro}
Consider a scenario where a group of voters needs to select a committee, that is a given-size subset of available candidates. Assume the voters have approval-based preferences: each voter submits a ballot in which she indicates which of the available candidates she finds acceptable. This scenario received a considerable attention in the literature in recent years---see the book chapter by \citeauthor{kil:chapter:approval-multiwinner}~\shortcite{kil:chapter:approval-multiwinner} and the recent survey by \citeauthor{lac-sko:abc-survey}~\shortcite{lac-sko:abc-survey}. 

In many scenarios that fit in the model of approval-based committee elections, for example when the goal is to select a representative body for a population of voters, it is required that the elected committee should represent the voters proportionally. Typically, the term proportionality is used to indicate that each group of voters with similar opinions should approve the number of elected candidates that is proportional to the size of the group. For example, consider a society that is divided into two coherent groups: there are two disjoint sets of candidates, $A$ and $B$; 60\% voters approve $A$ and 40\% approve $B$. Then, proportionality---in its most commonly used sense---means that we shall select roughly 60\% of committee members from $A$ and 40\% of committee members from $B$.

Is this outcome fair to the voters? That depends on how we define and interpret voters' satisfaction. If we define the satisfaction of a voter as the number of elected candidates she approves, then indeed the outcome seems fair. However, if the elected committee is to take a number of majoritarian decisions, then it is likely that such decisions will almost solely satisfy the voters from the first group, which can be considered highly unfair by those from the second group. Similar arguments led to the idea of \emph{degressive proportionality}~\cite{laslier2012WhyNotProportional,MT12,KLMT13a}, 
where it is advised that smaller groups of voters shall obtain the number of representatives that is greater than the linear proportionality would require. In fact, degressive proportional committees can be observed in real world---European Parliament is perhaps the most commonly known example of such a committee~\cite{rose2013RepresentingEuropeans}.\footnote{Degressive proportional rules have been also considered for selecting the United Nations Parliamentary Assembly~\cite{unAssembly} 
and for allocating weights in the Council of the European Union~\cite{euVoting}.
}

On the other hand, in certain applications one would prefer to use a voting rule that follows the opposite principle of regressive proportionality. As an example, consider a group of experts selecting grant proposals for funding. Here, we would like to use a rule that provides some degree of proportionality in order to ensure that different scientific disciplines are fairly represented among the selected project proposals. On the other hand, one would perhaps prefer to select mainly projects that obtain a large support from the experts, following the idea of regressive proportionality.

\subsection{Our Contribution}
In this paper we introduce a new family of rules that follow the principles of degressive and regressive proportionality, and analyse these rules taking three different viewpoints:
\begin{enumerate}[\hspace{0pt}1.]
    \item The worst-case approach: for a given election rule we ask what is the guaranteed number of representatives that a coherent group of voters who form a  $\gamma$-fraction of the society gets in the elected committee.
    \item The average-case approach: assuming that the voters and the candidates are represented as points in a one-dimensional issue space, we ask how different forms of proportionality map distributions of voters in the issue space to distributions of their satisfactions, measured as numbers of representatives in elected committees. 
    \item The analysis of the voting committee model~\cite{skow:multiwinner-models,ijcai2020-28}: 
    we ask how different forms of proportionality map distributions of voters in the issue space to distributions of their satisfactions from decisions made by elected committees. 
\end{enumerate}

Our rules are natural extensions of Sequential Phragm\'{e}n's Rule, which is known to behave very well in terms of proportionality~\cite{bri-las-sko:c:apportionment,aaai/BrillFJL17-phragmen,skowron:prop-degree,pet-sko:laminar}. These rules are practical, and can be computed in polynomial time. We also compare them with other voting methods, such as (sequential) Thiele rules.

\section{Preliminaries}

For each $n \in \naturals$ we set $[n] = \{1,\ldots,n\}$. For a set $X$ we use $S_k(X)$ to denote the set of all $k$-element subsets of $X$; by $S(X)$ we denote the set of nonempty subsets of $X$.

An approval-based election is a quadruple $(C, V, A, k)$, where $C=\{c_1, \ldots, c_m\}$ is a set of candidates, $V=\{v_1, \ldots, v_n\}$ is a set of voters, $A\colon V \to S(C)$ is a function that maps each voter to a subset of candidates that she approves---we call $A$ an \emph{approval-based profile}, and $k$ is a desired size of the committee to be elected. We use $n$ and $m$ to denote the number of voters and candidates, respectively, i.e., $|V|=n$ and $|C|=m$.

We call the elements of $S_k(C)$ \emph{size-$k$ committees}, or simply committees, if $k$ is clear from the context. An \emph{approval-based committee election rule}, in short a rule, is a function $\mathcal{R}$  that for each election instance $E=(C,V,A,k)$ returns one or multiple size-$k$ committees, i.e., $\mathcal{R}(E) \in S(S_k(C))$. We call elements of $\mathcal{R}(E)$ \emph{winning committees}. 

\section{The Class of $\alpha$/$\beta$-Phragm\'{e}n's Rules}

In this paper we introduce two classes of election rules that generalize the Phragm\'{e}n's sequential rule (for a broader discussion of the Phragm\'{e}n's rule we refer to the recent survey by \citeauthor{lac-sko:abc-survey}~\shortcite{lac-sko:abc-survey}). The definitions of those classes are based on the following idea. The voters earn virtual money---credits---over time (the time is continuous), and they use the credits they earned to pay for committee members that they approve; buying each candidate costs a certain amount of money.
The rules are sequential---they start with an empty committee $W = \emptyset$ and iteratively add candidates to $W$. The voters are greedy: in the first time moment, when there is a group of voters and a not-yet-selected candidate $c \notin W$, such that the voters who approve $c$ have altogether certain amount of unspent money, the rule stops, adds $c$ to the committee, asks the voters to pay for $c$ (resetting their credits to zeros), and resumes. The rule finishes when $k$ candidates are selected.

In the original Phragm\'{e}n's sequential rule the voters earn credits with a constant speed (e.g., one credit per time unit) and each candidate costs 1 credit. Here, we consider the following two variants of the rule: 
\begin{enumerate}
    \item We allow the speed of earning to change, dependently on the number of candidates that the voters like in the already assembled committee. Formally, let us fix a positive, non-increasing,  discrete function $\alpha: \naturals_+ \to (0, 1]$ such that $\alpha(1)=1$. In the $\alpha$-Phragm\'{e}n's rule, $\alpha(i)$ is the voter's speed of earning credits per time unit in case $(i-1)$ committee members are already approved by the voter. In other words, each time voter $v$ pays for the $i$-th candidate, her speed of earning credits changes to $\alpha(i+1)$ and remains the same until the next $v$'s purchase of a candidate. In particular, the original Phragm\'{e}n's sequential rule corresponds to the $\alpha$-Phragm\'{e}n's rule, for the constant sequence $\alpha$. In this paper we will consider $\alpha$-Phragm\'{e}n's rules for non-increasing functions $\alpha$.
    
    \item We allow the costs of the candidates to differ dependently on the number of voters who approve the specific candidate. Formally, let us fix a positive function $\beta: [0, 1] \to (0, 1]$ such that $\beta(0)=1$. In the $\beta$-Phragm\'{e}n's rule, $\beta(\nicefrac{|V'|}{|V|})$ is the cost of the candidate that is approved exactly by the voters from $V'$. In other words, the voters from $V'$ can buy candidate $c'$ that they approve, if they have $\beta(\nicefrac{|V'|}{|V|})$ credits in total. In particular, the original Phragm\'{e}n's sequential rule corresponds to the $\beta$-Phragm\'{e}n's rule, for the constant function $\beta$. In this paper we will consider $\beta$-Phragm\'{e}n's rules for non-increasing functions $\beta$.
\end{enumerate}

\Cref{ex:definition} below illustrates the procedure of selecting winning candidates according to the $\alpha$-Phragm\'{e}n's and $\beta$-Phragm\'{e}n's rules.

\begin{figure}
\begin{center}
\begin{tikzpicture}
[yscale=1.0,xscale=2.0]
\newcommand{\width}{.5}

\filldraw[fill=green!10!white, draw=black] (-\width/2,0) rectangle (-\width/2 + 4 * \width,0.5)
node[pos=.5] {$c_1$};
\filldraw[fill=green!30!white, draw=black] (-\width/2,0.5) rectangle (-\width/2 + 3 * \width,1.0)
node[pos=.5] {$c_4$};
\filldraw[fill=red!10!white, draw=black] (-\width/2+ 3 * \width,0.5) rectangle (-\width/2 + 5 * \width,1)
node[pos=.5] {$c_5$};
\filldraw[fill=red!30!white, draw=black] (-\width/2+ 6 * \width,0) rectangle (-\width/2 + 7 * \width,.5)
node[pos=.5] {$c_3$};
\filldraw[fill=cyan!10!white, draw=black] (-\width/2+ 4 * \width,0) rectangle (-\width/2 + 6 * \width,0.5)
node[pos=.5] {$c_2$};
\filldraw[fill=cyan!30!white, draw=black] (-\width/2+ 0 * \width,1.0) rectangle (-\width/2 + 2 * \width,1.5)
node[pos=.5] {$c_6$};
\filldraw[fill=cyan!30!white, draw=black] (-\width/2+ 3 * \width,1.0) rectangle (-\width/2 + 4 * \width,1.5)
node[pos=.5] {$c_6$};

\foreach \x in {1,...,7}
\node at (-\width + \x*\width, -0.35) {$v_{\x}$};

\draw [decorate,decoration={brace,amplitude=10pt},yshift=0pt]
(6.5*\width,-0.6) -- (-0.5*\width,-0.6) node [black,midway,below,yshift=-8pt] 
{voters};
\draw [decorate,decoration={brace,amplitude=10pt},yshift=0pt]
(-0.6*\width,-0.1) -- (-0.6*\width,1.5) node [black,midway,below,xshift=-27pt, rotate=90] 
{candidates};
\end{tikzpicture}
\vspace{-0.3cm}
\caption{Graphical representation of the approval profile used in \Cref{ex:definition}. Each voter approves the candidates that are above the voter in the figure. For instance, we have $A(v_4)=\{c_1, c_5, c_6\}$.}
\label{fig:example}
\end{center}
\vspace{-0.6cm}
\end{figure}
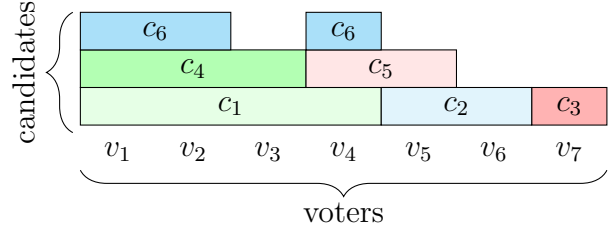

\begin{example}\label{ex:definition}
Consider an instance $E=(C, V, A, k)$ depicted in \Cref{fig:example} for $k=3$, and let us fix two functions: $\alpha(i)=(\nicefrac{1}{10})^{i-1}$ and $\beta(x)=(\nicefrac{1}{4})^{7x}$. The procedures of selecting candidates to the committee according to $\alpha$-Phragm\'{e}n's and $\beta$-Phragm\'{e}n's rules are as follows.

Let us first consider the $\alpha$-Phragm\'{e}n's rule, which starts with the empty set $W= \emptyset$. Each voter earns credits with the speed of $\alpha(1) = 1$ per time unit. At time $\nicefrac{1}{4}$ voters $\{v_1, v_2, v_3, v_4\}$ buy candidate $c_1$---these voters spend all so-far earned money on $c_1$ and their speed of earning changes to $\alpha(2) = \nicefrac{1}{10}$. Next, after $\nicefrac{1}{4}$ time units voters $\{v_5, v_6\}$ buy candidate $c_2$---let us check that voters $\{v_5,v_6\}$ can afford to buy $c_2$. Voters $v_5$ and $v_6$ have $\nicefrac{1}{4} + \nicefrac{1}{4} = \nicefrac{1}{2}$ credits each. Hence, in total they have one credit. At this moment: 
\begin{enumerate}[\hspace{0pt}1.]
    \item $W=\{c_1,c_2\}$;
    \item Each voter from $\{v_1, v_2, v_3, v_4\}$ has $\nicefrac{1}{40}$ credits and their speed of earning is $\alpha(2) = \nicefrac{1}{10}$;
    \item Voters from $\{v_5,v_6\}$ have 0 credits. They have one approved candidate in $W$, thus their speed of earning credits is $\alpha(2) = \nicefrac{1}{10}$;
    \item voter $v_7$ earns credits with the speed of $\alpha(1) = 1$, and she already has $\nicefrac{1}{2}$ credits.
\end{enumerate}
In the last step the rule selects $c_3$. Indeed, after $\nicefrac{1}{2}$ time units voter $v_7$ who approves candidate $c_3$ has one credit. What is more, within the considered time frame:
\begin{enumerate}[\hspace{0pt}1.]
    \item Each voter from $\{v_1, v_2, v_3, v_4\}$ has $\nicefrac{1}{40} + \nicefrac{1}{20} = \nicefrac{3}{40}$ credits;
    \item Each voter from $\{v_5,v_6\}$ has $\nicefrac{1}{20}$ credits.
\end{enumerate}
Based on the above observation, it is easy to see that no other candidate can be bought. Hence, $W=\{c_1, c_2, c_3\}$.

Now, let us consider the $\beta$-Phragm\'{e}n's rule. Again, we start with $W= \emptyset$. Each voter earns money with the speed of one credit per time unit. At time $\nicefrac{1}{1024}$ voters $\{v_1, v_2, v_3, v_4\}$ buy candidate $c_1$ (since candidate $c_1$ is approved by 4 voters, she costs $\beta(\nicefrac{4}{7})=\nicefrac{1}{256}$)---these voters spend all so-far earned money on $c_1$. Next, after $\nicefrac{1}{192}$ time units voters from $\{v_1, v_2, v_3\}$ can buy candidate $c_4$ and voters from $\{v_1, v_2, v_4\}$ can buy candidate $c_6$. On the other hand, let us check that voters $\{v_5, v_6\}$ cannot afford to buy $c_2$ in that time frame. Indeed, voters from $\{v_5, v_6\}$ have $2 \cdot \nicefrac{1}{1024} + 2 \cdot \nicefrac{1}{192} = \nicefrac{3+2^4}{3 \cdot 2^9}<\nicefrac{1}{16}$. 
Let us assume that candidate $c_6$ is bought by $\{v_1, v_2, v_4\}$ and finally after $\nicefrac{1}{288}$ time units voters from $\{v_1, v_2, v_3\}$ can buy candidate $c_4$---indeed, they have $\nicefrac{1}{192} + 3 \cdot \nicefrac{1}{288}=\nicefrac{1}{192} + \nicefrac{1}{96}=\nicefrac{3}{192}=\nicefrac{1}{64}$. One can also easily compute that other candidate cannot be bought in that time frame. Hence, the winning committee is $W=\{c_1, c_4, c_6\}$.
\end{example}

We will compare $\alpha$-Phragm\'{e}n's and $\beta$-Phragm\'{e}n's rules with $\lambda$-Thiele methods, another broad and important class of approval-based committee election rules. Given a non-increasing, convex function $\lambda: \naturals \to \reals$, a $\lambda$-Thiele method picks those committees $W$ that maximize the following quantity:
\begin{align*}
    \lambdascore (W) = \textstyle\sum_{v \in V} \sum_{j=1}^{|W \cap A(v)|}\lambda(j) \text{.}
\end{align*}

A particularly important example of a $\lambda$-Thiele method is Proportional Approval Voting (PAV), implemented by the harmonic sequence of weights $\lambda_{\pav}(i) = \nicefrac{1}{i}$. Given $q \leq 1$ we define the $q$-geometric-Thiele rule as the one that is implemented by $\lambda_{\geom}(i) = q^i$.

\begin{example}
Let us consider an instance $E=(C, V, A, k)$ depicted in \Cref{fig:example} for $k=3$. In the case of PAV rule ($\lambda_{PAV}$-Thiele rule), the winning committees are $W=\{c_1, c_2, c_4\}$ and $W' =\{c_1, c_2, c_6\}$ with the score of
\begin{align*}
    \lambdascore(W) = \lambdascore(W') = 3\cdot(1+\nicefrac{1}{2}) + 3\cdot1= \nicefrac{15}{2}.
\end{align*}
\end{example}

Another important class of approval-based committee election rules consists of  sequential variants of $\lambda$-Thiele methods. The seq-$\lambda$-Thiele rule selects candidates in rounds. Let $W_i$ denote the set of candidates picked until the end of the $i$-th round. In the $i$-th round the rule chooses a candidate $c$ that maximizes $\lambdascore(W \cup \{c\})$, and adds $c$ to the committee. 

\section{Proportionality of $\alpha$/$\beta$-Phragm\'{e}n's Rules}

In this section we assess how well committees returned by $\alpha$-Phragm\'{e}n's and $\beta$-Phragm\'{e}n's rules represent minorities of voters, depending on the sizes of these minorities.

In \Cref{def:pjr_guarantee} we formulate the axiom of proportional justified representation  degree (PJR degree), which is a quantitative variant of PJR~\cite{SFF+17a}. Similarly to PJR and other related properties, such as extended justified representation (EJR)~\cite{azi-bri-con-elk-fre-wal:c:justified-representation} or lower quota~\cite{bri-las-sko:c:apportionment,lac-sko:t:approval-thiele}, our axiom requires that groups of voters of sufficient sizes and with cohesive preferences should have right to decide about certain fractions of elected committees. However, since our goal is to analyze rules which are not proportional in the classic sense, our axiom does not have an encoded threshold specifying how many candidates cohesive groups of voters are allowed to elect. Instead, this threshold is provided as an adjustable function which allows to quantify the level to which the rule respects opinions of voters with cohesive preferences. This way, the axiom is more similar to proportionality degree~\cite{skowron:prop-degree}, a quantitative version of EJR~\cite{azi-bri-con-elk-fre-wal:c:justified-representation}.  Considering a quantitative variant of PJR rather than of EJR is motivated by the fact that the original Phragm\'{e}n's rule which we generalise in this paper, satisfies PJR and violates EJR. Thus, when analyzing the PJR degree of $\alpha$/$\beta$-Phragm\'{e}n's rules for other than constant $\alpha$/$\beta$-functions, we will have a reference point of a perfectly linearly-proportional rule in the class.


\begin{definition}[Proportional justified representation (PJR) degree]\label{def:pjr_guarantee}
Let $f\colon [0,1] \times \naturals \to \rationals$. We say that a rule $\calR$ has the PJR degree of $f$ if for each election instance $E = (C, V, A, k)$, each winning committee
$W \in \calR(E)$, and each group of voters $S \subset V$ it holds that:
\begin{align*}
\left|\bigcup_{v \in S} (A(v) \cap W) \right| \geq \min\left(\left| \bigcap_{v \in S}A(v) \right|, \Big\lfloor f(\nicefrac{|S|}{|V|}, k) \Big\rfloor\right) \text{.}
\end{align*}
\end{definition}

For example, if a group of voters $S$ form $20\%$ of the whole society and if it agrees on sufficiently many candidates ($\left| \bigcap_{v \in S}A(v) \right|$ is sufficiently high), then the PJR degree of $f$ means this group is allowed to decide at least about $f(0.2, k)$ members of the committee. If those voters agree on less than $f(0.2, k)$ candidates, the number of candidate they are allowed to decide about is truncated accordingly.  


Below, we present the main theoretical results for the $\alpha$-Phragm\'{e}n's rule.

\begin{theorem}\label{thm:alpha}
For a non-increasing, positive function $\alpha: \naturals_+ \to (0,1]$. The $\alpha$-Phragm\'{e}n's rule has the PJR degree of $f_\alpha$, where $f_\alpha(\gamma,k)$ is the largest natural number $\ell$ such that:
\begin{align*}
    \sum_{i=1}^{\ell}\nicefrac{1}{\alpha(i)} \leq (k-\ell+1)\frac{\gamma}{1-\gamma}.
\end{align*}
\end{theorem}

\begin{proof}
Observe that the $\alpha$-Phrgamen's rule is defined as a process of earning and spending money with certain constraints; these constraints, e.g., say that the voters make their purchases greedily. In order to assess the number of selected candidates, who are approved by the voters from a group $S$, we will analyse a procedure of buying candidates with relaxed constraints. For this simpler procedure we will be able to identify the case where the voters from $S$ have least money left at the end. Since our constraints are relaxed, this will give us lower bounds on the amount of money that the voters from $S$ are left with at the end of the execution of the $\alpha$-Phrgamen's rule. We will show that this amount of money is sufficient to buy an additional candidate unless certain number of candidates have been already bought by the voters from $S$.

Let us look at the following process of earning and spending money by the voters from~$S$. We first fix the time when the procedure stops, $t^*$.  The only constraint that we have is that the voters from $S$ cannot have more than one unit of money in the specific time (otherwise, they would buy an additional candidate). In that scheme we will assess the amount of money that the voters from $S$ have at time $t^*$. We aim to find instances such that the amount of credits left at $t^*$ is the smallest. 

At first, we start with an arbitrary instance of election $E$.
Let $z$ be the largest natural number such that in each instance, in time $t^*$ the voters from $S$ will buy at least $z$ candidates. 
In particular, in $E$ our process will select at least $z$ candidates who are approved by some voters in $S$. Let $t_1, t_2, \ldots, t_z$ be the time points in which some voters from $S$ are involved in some purchases of candidates;  for each $i\in \{2, \ldots, z\}$, $t_{i-1} \leq t_i$ and $t_z \leq t^*$. Without loss of generality, assume that in time $t_i$ candidate $c_i$ is bought. Our constraint says that the voters from $S$ pay in total at most one unit of money for $c_i$.
Now, we are going to modify the instance in such a way that after all such modifications the voters from $S$ have the least amount of money in time $t^*$. We can consider the minimisation problem with the following objective function
\begin{align}\label{proof:min:alpha}
    \sum_{i=1}^{|S|} (t^*-t_{v_i})\alpha(z_i),
\end{align}
where $t_{v_i}$ is the last time when $v_i$ made her last purchase, and $\alpha(z_i)$ is $v_i$'s speed of earning credits after buying her last candidate; in particular this means that voter $i$ approves $z_i-1$ selected candidates. Now, we observe that for each instance we can reduce the value of objective function by performing the following procedure.
\begin{enumerate}
    \item We modify each purchase of candidates by making all voters to pay for the specific candidate. Observe that this modification lowers the value of the objective function as for each $i\in \{1,\ldots, n\}$, $t_{v_i}=t_z$ and $\alpha(z_i)=\alpha(z+1)$. 
    Note that this modification will not violate the constraint of the purchase process since at each time $t_i$ the voters from $S$ have at most one dollar in total, thus they will all pay for each selected candidate at most one dollar. Otherwise, they would buy an additional candidate from the set on which they agree. 
    
    \item We observe that $t_z$ is the highest when (1) all the voters pay one credit for each candidate, and when (2) each candidate who is paid by some voters from $S$ is paid only by such voters. In that scenario we have
    \begin{align*}
        t_z = \nicefrac{1}{|S|}\sum_{i=1}^z \nicefrac{1}{\alpha(i)}
    \end{align*}
    and the objective function equals to
    \begin{align}\label{test}
        \sum_{i=1}^{|S|} (t^*-t_{v_i})\alpha(v_i) 
        = \sum_{i=1}^{|S|} (t^* - t_z)\alpha(z+1) 
        = \left(t^* - \nicefrac{1}{|S|}\sum_{i=1}^z \nicefrac{1}{\alpha(i)}\right)\alpha(z+1) |S|.
    \end{align}
\end{enumerate}
To sum up, the instances that solve the minimization problem given by the objective function \eqref{proof:min:alpha} are such that in each time of purchase all of the voters $S$ buy a candidate for the amount of one credit.

Now, we assess the lower bound on terminal time $t^*$ of the process of selecting $k$ candidates to the committee. If the voters from $S$ select $z$ candidates, then those from $V \setminus S$ are responsible for selecting $(k-z)$ candidates. Hence, they have to earn at least $(k-z)$ credits. Due to the assumption that the $\alpha$ function is non-increasing, we see that the voters from $V \setminus S$ need at least $t^*_{z} = \nicefrac{(k-z)}{|V \setminus S|}$ units of time to earn $(k - z)$ credits.

Let us assess what are the values of $\ell$ that have the following property: if the voters from $S$ buy $(\ell-1)$ candidates, they would be guaranteed to have more than 1 credit by the time the process stops and therefore they could buy one more candidate. Let  $\gamma = \nicefrac{|S|}{|V|}$ Using the above considered minimization problem with the form of objective function specified in \eqref{test}, we have the following inequalities
\begin{align*}
    1 &\leq \left(t^*_{\ell -1}-\frac{1}{|S|}\sum_{i=1}^{\ell-1}\frac{1}{\alpha(i)}\right)\alpha(\ell)|S| \implies \\
    1 &\leq \left(\frac{k-\ell+1}{|V|-|S|}-\frac{1}{|S|}\sum_{i=1}^{\ell-1}\frac{1}{\alpha(i)}\right)\alpha(\ell)|S| \implies \\
    \sum_{i=1}^\ell\nicefrac{1}{\alpha(i)} &\leq (k-\ell+1)\frac{\gamma}{1-\gamma}.
\end{align*}
Note that if the above inequality is satisfied for a given value of $\ell$, then it is also satisfied for lower values of $\ell$. Thus, the largest value of $\ell$ that satisfies the above inequality guarantees that the voters from $S$ cannot buy less than $\ell$ candidates.
\end{proof}

The results presented in the \Cref{thm:alpha} are general, but hard to interpret. In order to simplify the expression presented in the \Cref{thm:alpha} we observe that $\sum_{i=1}^{\ell}\nicefrac{1}{\alpha(i)} \geq \ell$, and formulate the following corollary.

\begin{corollary}\label{cor:alpha}
Fix a non-increasing, positive function $\alpha: \naturals_+ \to (0,1]$. The $\alpha$-Phragm\'{e}n's rule has the PJR degree of $f_\alpha$, where $f_\alpha(\gamma, k)$ is the largest natural number $\ell$ such that:
\begin{align*}
    \sum_{i=1}^{\ell}\nicefrac{1}{\alpha(i)} \leq \gamma(k+1).
\end{align*}
\end{corollary}

\Cref{cor:alpha} provides bounds which are easier to interpret, but weaker than those given in \Cref{thm:alpha}. Nevertheless, even \Cref{cor:alpha} itself implies that the classic Phragm\'{e}n's rule, which corresponds to the $\alpha$-Phragm\'{e}n's with the constant function $\alpha(i) = 1$, has the PJR degree of $f(\gamma, k) = \gamma(k+1) > \gamma \cdot k$, which implies satisfying PJR. 

We will use the simplified expression from \Cref{cor:alpha} to obtain PJR degree for specific classes of geometric Phragm\'{e}n rules with $q<1$.

\begin{proposition}\label{prop:geometric}
For a function $\alpha(i)=q^{i-1}$, where $q<1$, the $\alpha$-Phragm\'{e}n's rule has the PJR degree of $f_\alpha$:
    \begin{align*}
        f_\alpha(\gamma, k) = \lfloor\log_{\nicefrac{1}{q}} [\gamma(k+1)(\nicefrac{1}{q} - 1) + 1] - 1\rfloor.
    \end{align*}
\end{proposition}

\begin{proof}
We use the \Cref{cor:alpha}---we have:
\begin{align*}
    & \sum_{i=1}^{\ell}\nicefrac{1}{\alpha(i)} \leq \gamma(k+1) \implies 
    \frac{1 - \nicefrac{1}{q^{l+1}}}{1 - \nicefrac{1}{q}} \leq \gamma(k+1) \implies \\
    & 1 - \nicefrac{1}{q^{l+1}} \geq \gamma(k+1)(1 - \nicefrac{1}{q}) \implies 
     \nicefrac{1}{q^{l+1}} \leq \gamma(k+1)(\nicefrac{1}{q} - 1) + 1 \implies \\
    & \log_{\nicefrac{1}{q}} \nicefrac{1}{q}^{\ell+1} \leq \log_{\nicefrac{1}{q}} [\gamma(k+1)(\nicefrac{1}{q} - 1) + 1] \implies 
    \ell \leq \log_{\nicefrac{1}{q}} [\gamma(k+1)(\nicefrac{1}{q} - 1) + 1] - 1.
\end{align*}
Hence, the PJR degree $f_\alpha(\gamma, k)$ is given by
\begin{align*}
    \lfloor\log_{\nicefrac{1}{q}} [\gamma(k+1)(\nicefrac{1}{q} - 1) + 1] - 1\rfloor.
\end{align*}
\end{proof}

\begin{figure}[!t]
\begin{center}

\minipage{0.5\textwidth}
  \centering
  \includegraphics[width=\linewidth]{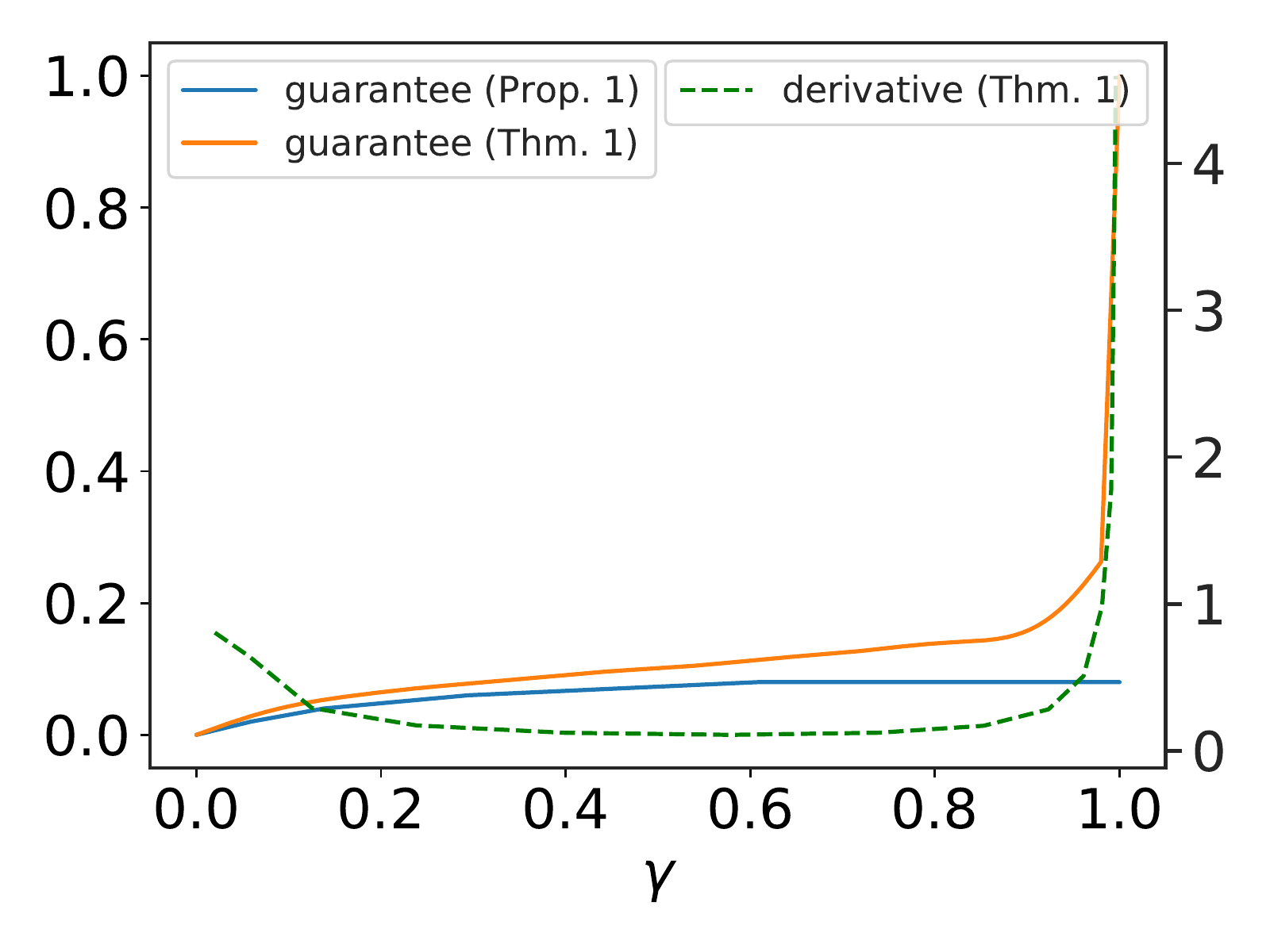}
  (a) $\alpha(i)=0.5^i; k=50$
\endminipage\hfill
\minipage{0.5\textwidth}
  \centering
  \includegraphics[width=\linewidth]{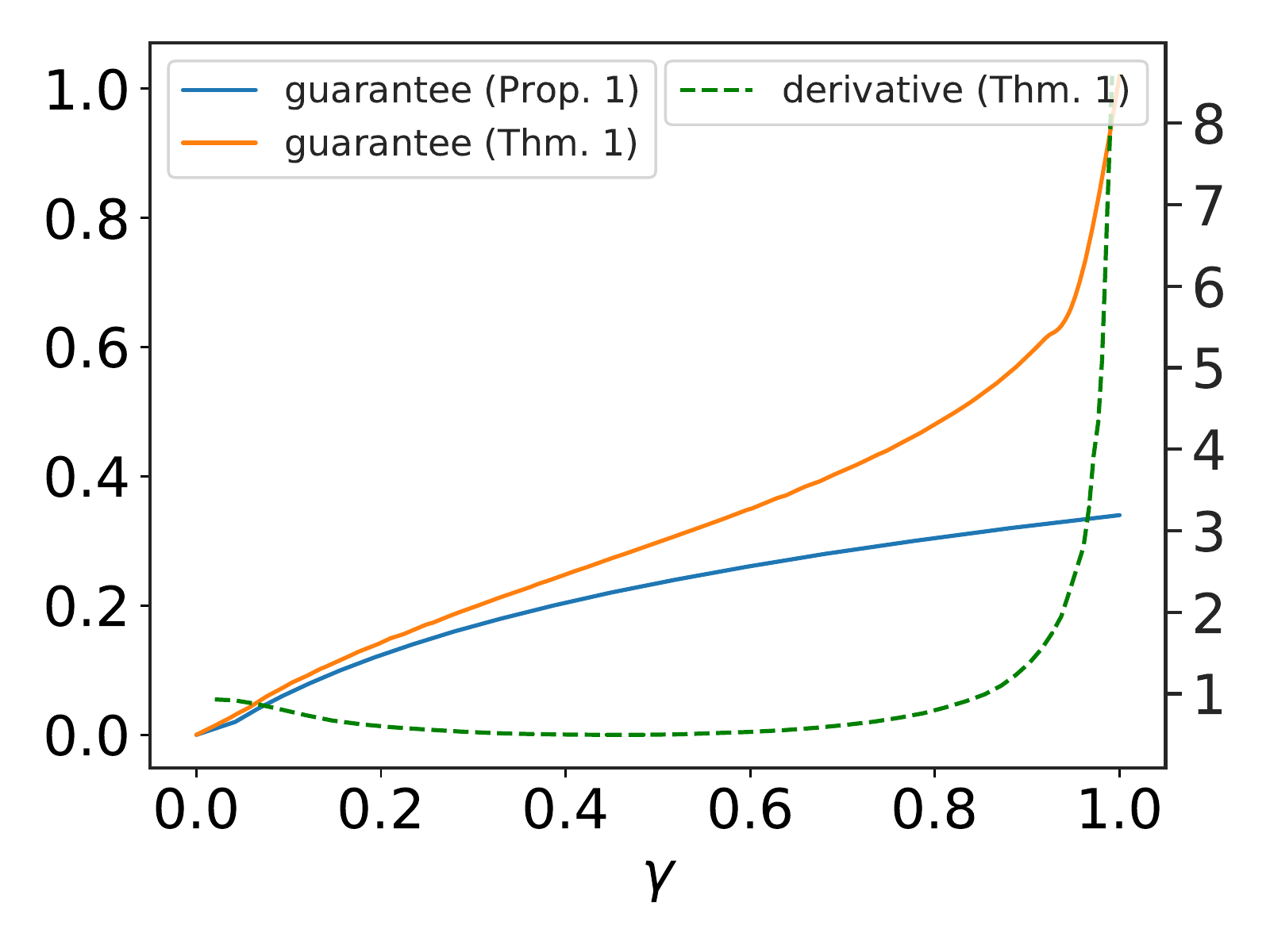}
  (b) $\alpha(i)=0.9^i; k=50$
\endminipage

\end{center}
\vspace{-0.3cm}
\caption{Lower bounds on the PJR degree (presented as a fraction of $k$) derived via \Cref{thm:alpha} and \Cref{prop:geometric} for two geometric $\alpha$-Phragm\'{e}n's rules. The right scale in $y$-axis is only relevant for the third line, which is the derivative of the guarantee derived via \Cref{thm:alpha}.}\label{fig:pjr_degree_alpha}
\end{figure}

The comparison of the results implied by \Cref{thm:alpha} and \Cref{prop:geometric} for two concrete examples of geometric $\alpha$-Phragm\'{e}n's rules are depicted in \Cref{fig:pjr_degree_alpha}. 
This figure leads to an interesting interpretation of what is degressive proportionality. At first one could expect that the PJR degree of a degressive proportional rule should have a plot that for small values of $\gamma$ lies above the plot for the linear function $g(\gamma) = \gamma$. Somehow surprisingly, this is not the case and the worst-case guarantee of each group is worse than in case of linear proportionality. The intuitive reason is that for each small cohesive group of voters there can always appear multiple groups which are even smaller and which should be even more privileged. \Cref{fig:pjr_degree_alpha} quantifies this effect in the worst-case. On the other hand, degressive proportionality means that for small values of $\gamma$ the ratio the PJR degree between a larger and a smaller group of voters is sublinear compared to the ratio of the sizes of the groups. This is visible by observing that for small values of $\gamma$ (here $\gamma \leq \nicefrac{1}{2}$) the derivative of the PJR degree is decreasing. We infer that one of the distinctive properties of the degressive proportionality is that the derivative of the proportionality guarantee is convex.  

Let us now we move to the analysis of $\beta$-Phragm\'{e}ns' rules.
\begin{theorem}\label{thm:beta}
Fix a non-increasing, positive function $\beta:[0,1] \to (0,1]$ with $\beta(0)=1$. The $\beta$-Phragm\'{e}n's rule has the PJR degree of $f_\beta$, where
\begin{align*}
    f_\beta(\gamma, k) = \Big\lfloor (k+1)\frac{\gamma\beta(1-\gamma)}{(1-\gamma)\beta(\gamma) + \gamma \beta(1-\gamma)}\Big\rfloor
\end{align*}
\end{theorem}

\begin{proof}
We will follow a similar proof technique to the one used in the proof of \Cref{thm:alpha}. 
Let us fix a group of voters $S$, and let $\gamma = \nicefrac{|S|}{|V|}$. Consider an election instance $E$. Assume the rule applied to $E$ ends in time $t^*$ and that the voters from $S$ approve $z < \min \left(\left|\bigcap_{v \in S}A(v) \right|, \left\lfloor f(\gamma, k) \right\rfloor\right)$ distinct elected candidates. First observe that at each time moment there exists a not elected candidate who is approved by all the voters from $S$. Thus, the voters from $S$ at each time cannot have more than $\beta(\gamma)$ unspent money. This means that at each time they paid for a candidate no more than $\beta(\gamma)$. Since they paid for $z$ candidates they paid in total no more than $z\beta(\gamma)$ dollars.

In time $t^*$ the total amount of money earned by the voters from $S$ equals $|S|t^*$. At that time the voters cannot be left with $\beta(\gamma)$ dollars or more, since then they would buy an additional candidate they all approve of. Thus
\begin{align*}
    |S|t^* - z\beta(\gamma) < \beta(\gamma)  \quad \text{and so~} t^* < \frac{z+1}{|S|}\beta(\gamma) \text{.} 
\end{align*}



On the other hand, we can observe that the voters from $V \backslash S$ have to earn for at least $(k-z)$ candidates on their own. The fastest way to do it is when they have $(k-z)$ candidates in common. Consequently:
\begin{align*}
    t^* \geq (k-z)\beta(1-\gamma)\nicefrac{1}{(|V|-|S|)}.
\end{align*}
Thus, it must hold that 
\begin{align}\label{eq:betaPhrgaCond}
   \frac{z+1}{|S|}\beta(\gamma) > (k-z)\beta(1-\gamma)\nicefrac{1}{(|V|-|S|)}.
\end{align}

In other words, if the above inequality did not hold, we would reach a contradiction with the initial assumption that $z < \min \left(\left|\bigcap_{v \in S}A(v) \right|, \left\lfloor f(\gamma, k) \right\rfloor\right)$. Thus, the voters from $S$ are guaranteed to approve $\ell$ candidates, where $\ell$ is the largest natural number such that \eqref{eq:betaPhrgaCond} would not hold for $z = \ell-1$:
\begin{align*}
   \frac{\ell}{|S|}\beta(\gamma) \leq (k-\ell-1)\beta(1-\gamma)\nicefrac{1}{(|V|-|S|)}.
\end{align*}
After algebraic reformulations:
\begin{align*}
    \ell &\leq (k-\ell+1)\frac{|S|\beta(1-\gamma)}{(|V|-|S|)\beta(\gamma)} \implies \\
    \ell\Big(1+\frac{\beta(1-\gamma)}{(\nicefrac{1}{\gamma}-1)\beta(\gamma)}\Big) &\leq (k+1)\frac{\beta(1-\gamma)}{(\nicefrac{1}{\gamma}-1)\beta(\gamma)} \implies \\
    \ell\Big((1-\gamma)\beta(\gamma)+\gamma\beta(1-\gamma)\Big) &\leq (k+1)\gamma\beta(1-\gamma) \implies \\
    \ell &\leq (k+1)\frac{\gamma\beta(1-\gamma)}{(1-\gamma)\beta(\gamma) + \gamma\beta(1-\gamma)}.
\end{align*}
The final statement is obtained by taking the floor from the right-hand side of above expression.
\end{proof}

Using the assumptions that $\beta$ is non-increasing, (for $\gamma \geq \nicefrac{1}{2}$ we have $\beta(\gamma)\leq \beta(1-\gamma)$ and for $\gamma \leq \nicefrac{1}{2}$ we have $\beta(\gamma)\geq \beta(1-\gamma)$) we obtained the following simplified version of the bounds.

\begin{corollary}\label{cor:beta}
The lower bound for PJR degree, $f_\beta(\gamma, k)$, satisfies the following:
\begin{enumerate}
    \item for $\gamma \geq \nicefrac{1}{2}$ we have $f_\beta(\gamma, k) \geq \lfloor (k+1)\gamma \rfloor$.
    \item for $\gamma \leq \nicefrac{1}{2}$ we have $f_\beta(\gamma, k) \geq \lfloor (k+1)\gamma \cdot \nicefrac{\beta(1-\gamma)}{\beta(\gamma)} \rfloor$.
\end{enumerate}
\end{corollary}

One could expect that the separated guarantee for regressive-proportional rules should be a function that is below $f(x) = x$ for small arguments and above $f(x) = x$ for large arguments. This is indeed the case, as illustrated in \Cref{fig:pjr_degree_beta}; the exact shape of the function quantify this effect. Similarly, according to our intuition the derivative of the PJR degree for regressive proportional rules is concave.

\begin{figure}[!t]
\begin{center}

\minipage{0.5\textwidth}
  \centering
  \includegraphics[width=\linewidth]{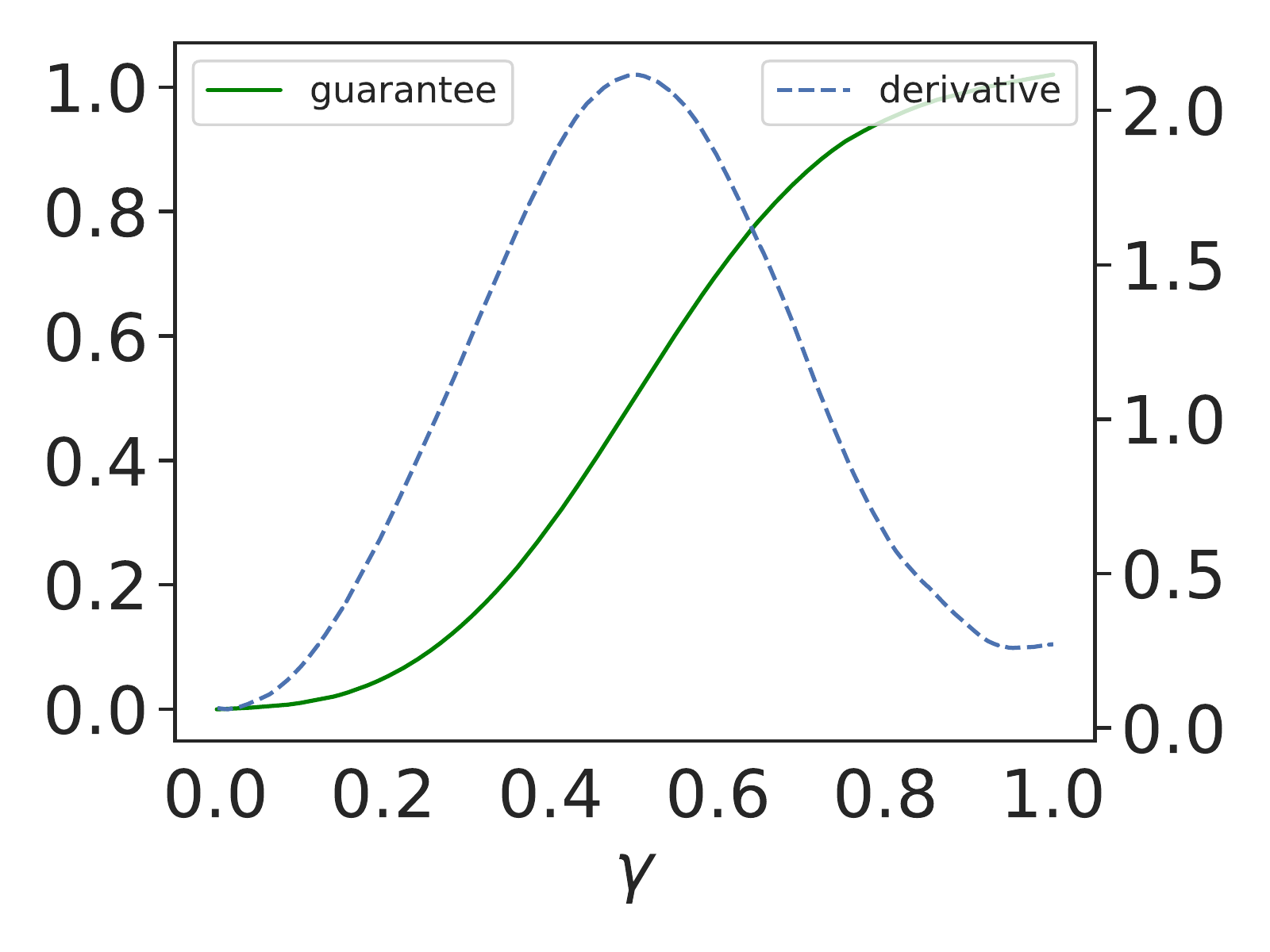}
  (a) $\beta(\gamma) = 0.1^{\gamma}; k=50$
\endminipage\hfill
\minipage{0.5\textwidth}
  \centering
  \includegraphics[width=\linewidth]{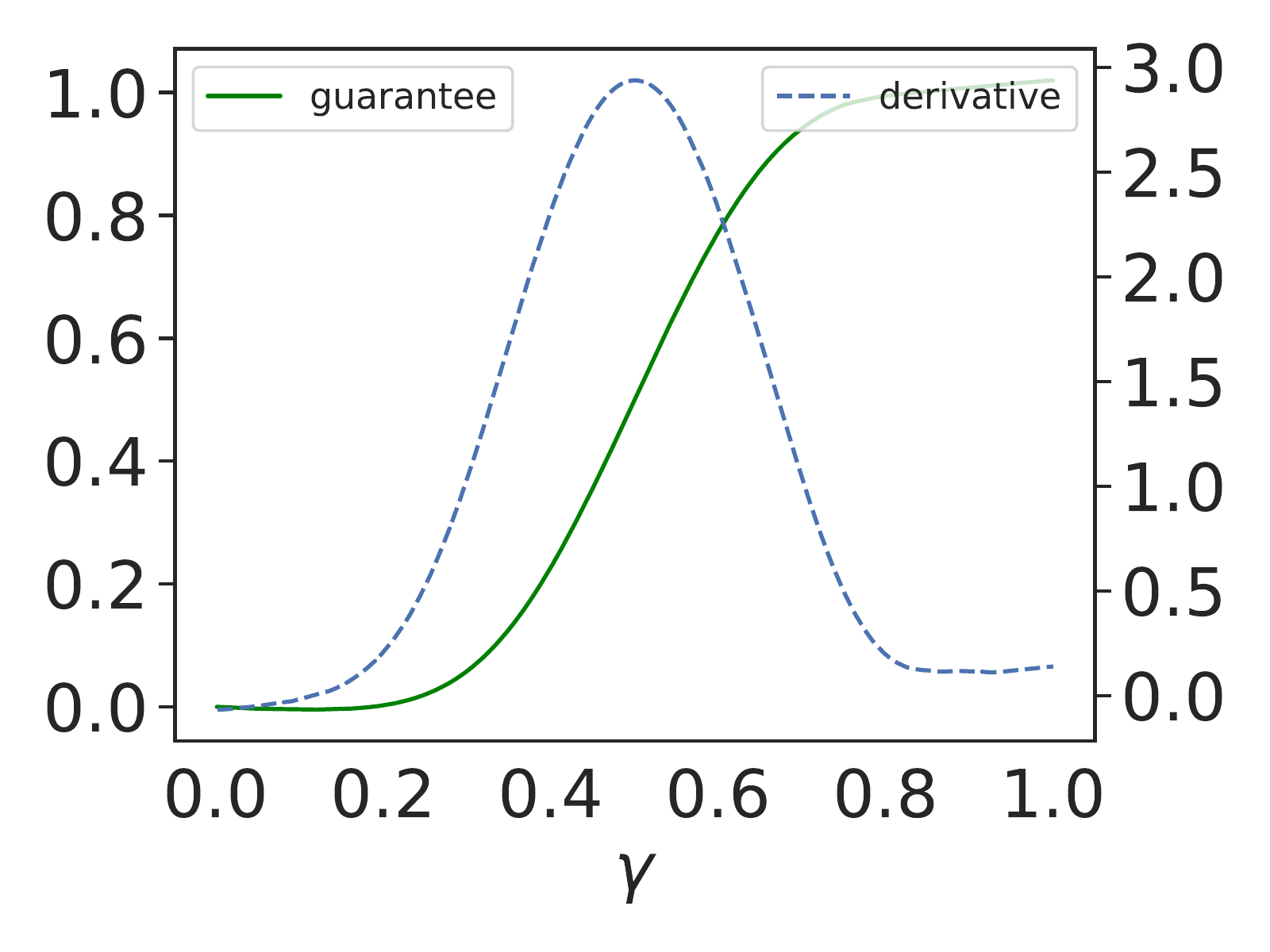}
  (b) $\beta(\gamma) = 0.01^{\gamma}; k=50$
\endminipage

\end{center}
\vspace{-0.3cm}
\caption{Lower bounds on the PJR degree (presented as a fraction of $k$) obtained from \Cref{thm:beta} for two $\beta$-Phragm\'{e}n's rules. The right scale in $y$-axis is relevant for the second line, which is the derivative of the PJR guarantee.}\label{fig:pjr_degree_beta}
\end{figure}

One can naturally ask: what if the speeds of earning money in the definition of $\alpha$-Phragm\'{e}n's rules are increasing? Will we obtain a regressive proportional rule? Interestingly, this is not the case, which is illustrated in the following example.

\begin{example}\label{ex:alpha_non_increasing}
Consider an $\alpha$-Phragm\'{e}n's rule with $\alpha(i) = i^{100}$, and the following election instance. There are 100 voters. Voters $v_1, \ldots v_{55}$ approve candidate $c_1$. Additionally, voters  $v_1, \ldots v_{30}$ approve $c_2, c_3, \ldots, c_6$. Further, voters $v_{51}, \ldots v_{100}$ approve $c_7, c_8, \ldots, c_{13}$. The size of the committee to be elected is $k = 6$. Here, the $\alpha$-Phragm\'{e}n's rule would select $c_1, \ldots, c_6$. Thus, the candidates who are approved by 50 voters, $c_7, c_8, \ldots, c_{13}$, would not be selected even though the voters who approved such candidates got only one or zero representatives. Instead, the rule would pick candidates who are approved by only 30 voters. This is not consistent with our interpretation of regressive proportionality. In contrary, from \Cref{cor:beta} it follows that each $\beta$-Phragm\'{e}n's rule would guarantee at least 3 candidates in the committee from those that are approved by 50 voters.
\end{example}

\Cref{ex:alpha_non_increasing} is very instructive. It illustrates that designing voting rules based solely on the intuitive premises can have undesirable effects. This illuminates the need of applying formal methods to the analysis of voting rules. Indeed, for an increasing function $\alpha$ the $\alpha$-Phragm\'{e}n's rule does not have a good PJR degree. In fact, exactly this observation has lead us to the definition of the class of $\beta$-Phragm\'{e}n's rules.  


\section{Comparing $\alpha/\beta$-Phragm\'{e}n's and $\lambda$-Thiele rules}

We will now compare the class of $\alpha/\beta$-Phragm\'{e}n's rules with the class of $\lambda$-Thiele methods. 

\subsection{The PJR degree of $\lambda$-Thiele rules}

We first observe that the lower bound on the PJR degree of Thiele methods follows from the analogous lower bound on the proportionality degree~\cite{skowron:prop-degree}. Since proportionality degree is a stronger condition than the PJR degree, we obtain the following corollary. 

\begin{theorem}[\cite{skowron:prop-degree}, Theorem~5.1]\label{thm:thiele_lower}
Let $\lambda: \naturals \to \reals$ be a non-increasing, convex function, and let $f_\lambda:\reals \times \naturals \to \reals$ such that $f_\lambda(\gamma,k) \leq k$, and that for each $x \in [k]$:
\begin{align*}
    (k-f_\lambda(\gamma, k))\lambda(1+f_\lambda(\gamma, k)) \geq \frac{1-\gamma}{\gamma} x\lambda(x).
\end{align*}
Then, the $\lambda$-Thiele rule has the PJR degree of $f_\lambda$.
\end{theorem}

We can also formulate the corresponding upper-bound on the PJR degree. We have: 

\begin{proposition}\label{prop:thiele_upper}
Let $\lambda: \naturals \to \reals$ be a non-increasing, convex function. The PJR degree of the $\lambda$-Thiele rule $f_\lambda$ must satisfy for each $x \in [k-f_\lambda(\gamma, k)+1]$:
\begin{align*}
(k - f_\lambda(\gamma, k) + x + 1) \lambda(f_\lambda(\gamma, k)) \geq \frac{1-\gamma}{\gamma} x\lambda(x) \text{.}
\end{align*}
\end{proposition}

\begin{proof}
For the sake of contradiction, let us assume that for some $\gamma \in (0,1)$, and  $x \in [k-f_\lambda(\gamma, k)+1]$ it holds that:
\begin{align*}
 (k - z + x + 1) \lambda(z) < \frac{1-\gamma}{\gamma} x\lambda(x) \text{,}
\end{align*}
where for the simplicity of notation we set $z=f_\lambda(\gamma, k)$.

We will construct an instance of an election witnessing that $f_\lambda$ cannot be a proportionality guarantee of the $\lambda$-Thiele rule. Let $C = B \cup D$ be the set of candidates, where $B=\{b_1, b_2, \ldots, b_z\}$ and $D=\{d_1, d_2, \ldots, d_{k-z+1}\}$. We distinguish two groups of voters, $S$ and $V \setminus S$, such that $\nicefrac{|S|}{|V|}=\gamma$. Each candidate $b\in B$ is approved by the voters from~$S$. Voters from $V \setminus S$ are divided into $\left\lceil\frac{k-z+1}{x}\right\rceil$ equally-sized groups: each group approves at most $x$ candidates from $D$, and each two groups approve disjoint sets of candidates. 

We will show that the optimal committee for this instance cannot contain $z$ candidates from~$B$. For the sake of contradiction, let us assume that an optimal committee $W$ contains all $z$ candidates from $B$. Then, the voters from $V \setminus S$ would have $(k-z)$ candidates that they approve in the winning committee~$W$. If we replaced one candidate from $B$ with a candidate from $D$ in $W$, then the score of 
the committee could not increase. Thus:
\begin{align*}
    |S|\lambda(z) \geq \frac{|V \setminus S|}{\left\lceil\frac{k-z+1}{x}\right\rceil} \lambda(x) \geq \frac{|V \setminus S|}{\frac{k-z+1}{x} + 1} \lambda(x) = \frac{|V \setminus S|}{k-z + x + 1} x\lambda(x)
\end{align*}
Consequently, we get that $(k - z + x + 1)\lambda(z) \geq \frac{1-\gamma}{\gamma} x\lambda(x)$, and so:
\begin{align*}
    (k - f_\lambda(\gamma, k) + x + 1) \lambda(f_\lambda(\gamma, k)) \geq \frac{1-\gamma}{\gamma} x\lambda(x) \text{.}
\end{align*}

Hence, we get the contradiction. Consequently, we get that set of candidates $D$ is included in the winning committee $W$ and the number of included candidates from $B$ is lower than $z$. Thus, $f_\lambda$ is not a PJR degree for the $\lambda$-Thiele method. This gives a contradiction, and completes the proof.
\end{proof}

\Cref{prop:thiele_upper} applies also to sequential Thiele methods. What is more, \Cref{thm:thiele_lower} says that the $q$-geometric Thiele method has the PJR degree of $f(\gamma, k) = z$, where $z$ is the largest value satisfying:
\begin{align*}
    (k-z)q^{1+z} \geq \frac{1-\gamma}{\gamma} \max_{x \in [k]}xq^x.
\end{align*}
For $q = 0.5, 0.8$ we depict this guarantee in \Cref{fig:separated_guarantee_degr_thiele}. We can see that the lower and upper bounds implied by \Cref{thm:thiele_lower} and \Cref{prop:thiele_upper} are very close. This shows that the analysis is almost tight.

\begin{figure}[!t]
\begin{center}

\minipage{0.5\textwidth}
  \centering
  \includegraphics[width=\linewidth]{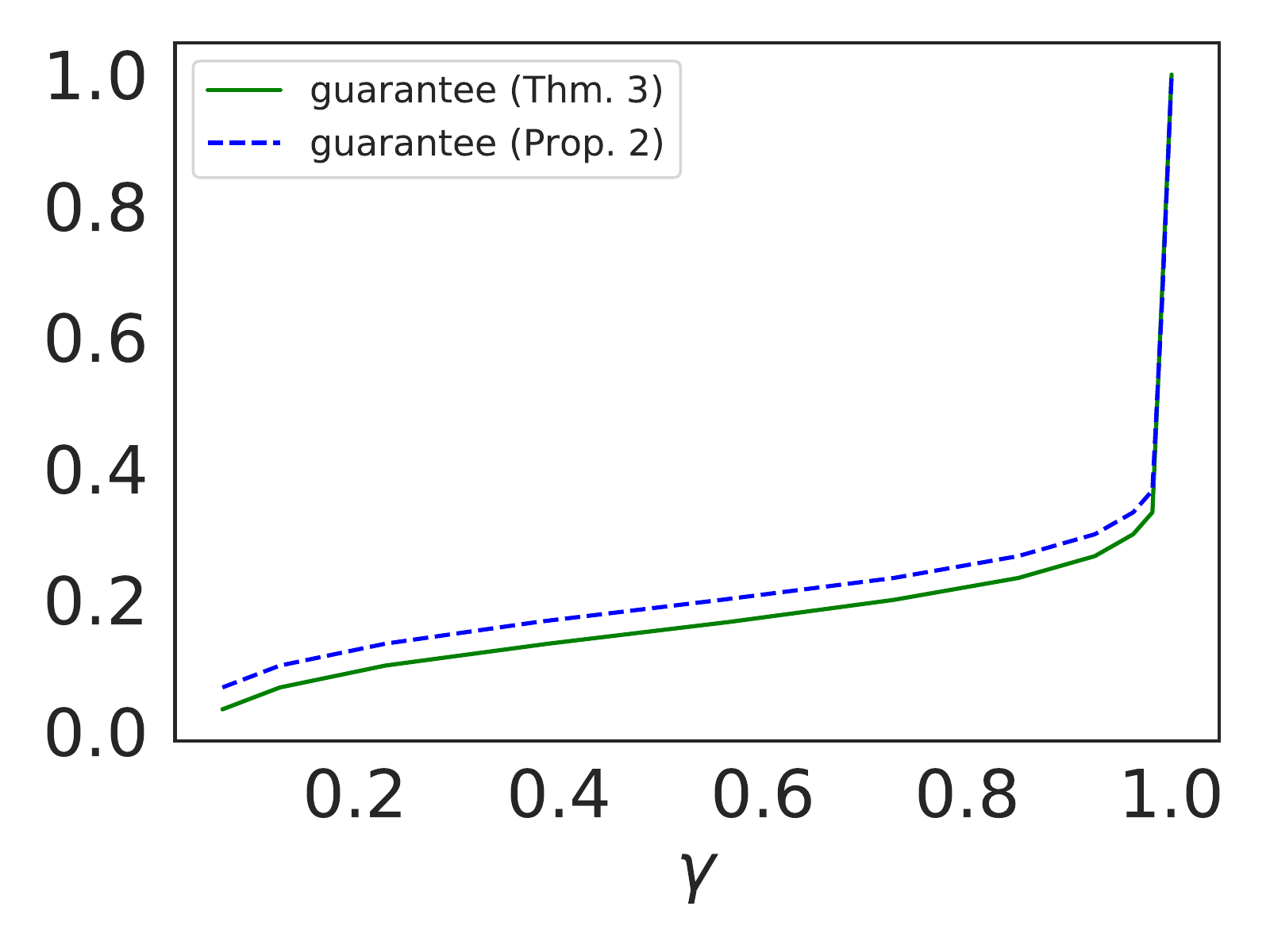}
  (a) $\lambda(\gamma) = 0.5^\gamma; k=30$
\endminipage\hfill
\minipage{0.5\textwidth}
  \centering
  \includegraphics[width=\linewidth]{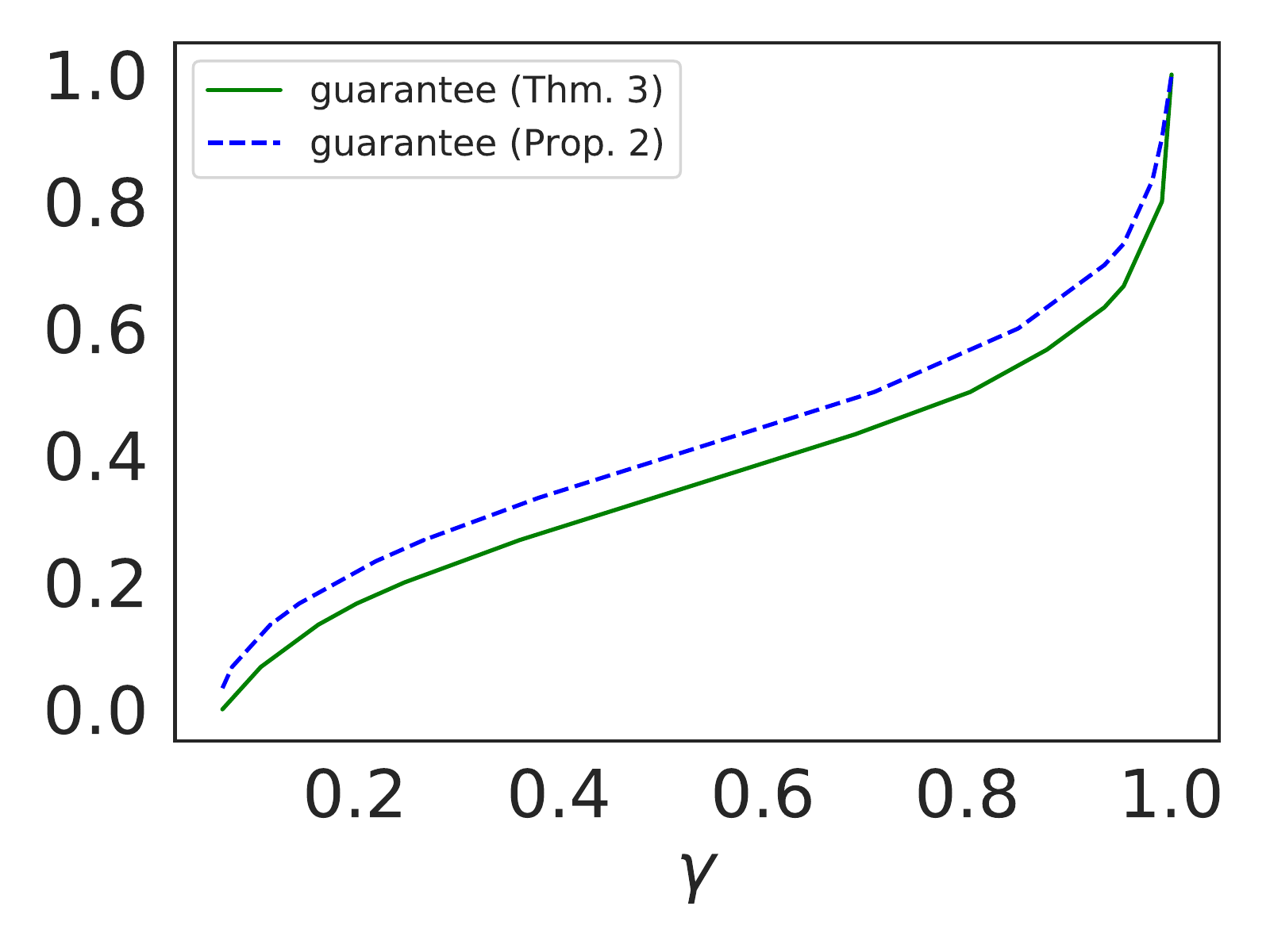}
  (b) $\lambda(\gamma) = 0.8^\gamma; k=30$
\endminipage

\end{center}
\vspace{-0.3cm}
\caption{Lower and upper bounds for the PJR degree (presented as a fraction of $k$) obtained from \Cref{thm:thiele_lower} and \Cref{prop:thiele_upper} for two geometric $\lambda$-Thiele rules, where $\lambda(\gamma)=q^{\gamma}$.}\label{fig:separated_guarantee_degr_thiele}
\vspace{-0.3cm}
\end{figure}

\subsection{Differences Between $\alpha/\beta$-Phragm\'{e}n's and $\lambda$-Thiele rules}

We can see that both $\alpha/\beta$-Phragm\'{e}n's and $\lambda$-Thiele rules spread the spectrum from degressive to regressive proportionality. In the remainder of this section we will explain that these classes are in fact different, and that certain behavior of $\alpha/\beta$-Phragm\'{e}n's rules cannot be implemented within the class of Thiele rules.

First, we observe that $\alpha/\beta$-Phragm\'{e}n's rules are computable in polynomial time, while for most $\lambda$-function computing the outcomes of $\lambda$-Thiele rules is $\np$-hard~\cite{sko-fal-lan:c:collective}.

Second, we note that $\alpha/\beta$-Phragm\'{e}n's rules satisfy the axiom of committee monotonicity~\cite{elk-fal-sko-sli:j:multiwinner-properties}. Intuitively, committee monotonicity says that if we increase the committee size, then the candidates that were members of winning committees should still be selected. This property is specifically important in certain contexts, where the goal is to produce a ranking of objects, so that the ranking proportionally reflect the opinions of a certain group of agents (see the work of \citeauthor{skowron:prop-degree}~\shortcite{{skowron:prop-degree}} for a more detailed discussion on the applications of committee monotonic approval-based rules). 


This property is also satisfied by sequential $\lambda$-Thiele rules. On the other hand, except for Approval Voting, no $\lambda$-Thiele method satisfies the property.
Thus, with respect to computation complexity and committee monotonicity, $\alpha/\beta$-Phragm\'{e}n's rules are closer to the sequential variants of $\lambda$-Thiele methods rather than to the $\lambda$-Thiele methods themselves.

Yet, there are properties of $\alpha/\beta$-Phragm\'{e}n's rules which cannot be satisfied by (sequential) $\lambda$-Thiele methods. In order to see that we first introduce an axiom of independence of unanimously approved candidate (IUAC).

\begin{definition}
Given an election $E = (C, V, A, k)$ and $c \in C$, by $(C_{-c}, V, A, k)$ we denote the instance obtained from $E$ by removing $c$ from the set of available candidates, and the approval sets of the voters.
A rule $\calR$ satisfies independence of unanimously approved candidate (IUAC), if for each election $E = (C, V, A, k)$ where there exists a single unanimously approved candidate $c$ (i.e., a candidate $c$ such that $c \in A(v)$ for all $v \in V$), for each winning committee $W \in \calR(E)$ there is $W' \in \calR((C_{-c}, V, A, k))$ such that $W = W' \cup \{c\}$.
\end{definition}

In words, assume we add one candidate who is approved by all the voters and that we increase the committee size by one; then IUAC requires that the rule should select the old committee together with the unanimously approved candidate.

We observe that that $\alpha$-Phragm\'{e}n's rules for geometric sequences $\alpha$ and every $\beta$-Phragm\'{e}n's rule satisfy the axiom of IUAC. Analogously, one can show that only $\lambda$-Thiele methods with geometric sequences of weights satisfy the axiom. 

\begin{proposition}\label{prop:thiele_geometric}
A $\lambda$-Thiele method satisfies IUAC if and only if $\lambda$ is a geometric sequence.
\end{proposition}

\begin{proof}
The proof that each geometric $\lambda$-Thiele method satisfies IUAC is straightforward. Now, let us fix a $\lambda$-Thiele method that satisfies IUAC and assume $\lambda(1)=1$. Let $q=\nicefrac{\lambda(2)}{\lambda(1)}$. If $\lambda$ is not geometric, then there exists $i$ such that $\lambda(i) \neq q^{i-1}$ (without loss of generality, we can assume that $\lambda(i) < q^{i-1}$; for $\lambda(i) > q^{i-1}$ the proof is analogous). Let $k = i-1$, and consider an instance where there are two disjoint groups of candidates, $C_1$ and $C_2$, and two disjoint groups of voters, $V_1$ and $V_2$, such that for each $v\in V_i$ we have $A(v)=C_i$. We set the sizes of the groups in such a way that $|V_2|=\frac{\lambda(k+1)+q^{k}}{2q}|V_1|$. First, we see that $|V_1|q^{k-1} > |V_2|$ which implies that only candidates from $C_1$ are selected to the winning committee. What is more, $q|V_2| > \lambda(k+1)|V_1|$, which implies that if we add new candidate $c$ that is approved by all the voters, then one candidate from $C_2$ would be selected to the committee of size $k+1$ which contradicts IUAC.
\end{proof}

Now, observe that there exists no geometric $\lambda$-Thiele method that implements regressive proportionality.


\begin{proposition}\label{prop:arbitrarily_small}
Fix a $q$-geometric Thiele method with $q<1$, $\gamma \in (0, 1)$ and $\varepsilon > 0$. Consider a cohesive group of voters $V'$ who form the $\gamma$ fraction of the whole society.
There is an instance such that the fraction of winning candidates that are approved by the voters from $V'$ is lower than $\varepsilon$.
\end{proposition}

\begin{proof}
 Consider the following instance. We set $C=C'\cup C_1 \cup \cdots \cup C_{k - \lfloor \varepsilon k \rfloor}$ and $V=V'\cup V_1 \cup \cdots \cup V_{k - \lfloor \varepsilon k \rfloor}$, where for each $i$, $|V_i|=(1 - \gamma)\nicefrac{1}{k - \lfloor \varepsilon k \rfloor+1}$. We assume that in each set $C_i$ there is a single candidate that is approved only by the voters from $V_i$. Additionally, the voters from $V'$ approve the candidates from $C'$. Thus, if a candidate from $C_i$ is selected to the committee, the voters from $V_i$ contribute together $(1 - \gamma)\nicefrac{1}{k - \lfloor \varepsilon k \rfloor+1}$ to the score. Now, we find $k$ that satisfies the following inequality
\begin{align*}
    (1 - \gamma) \frac{1}{k - \lfloor \varepsilon k \rfloor+1} > \gamma q^{\lfloor \varepsilon k \rfloor}.
\end{align*}
This implies that we select at most $\lfloor \varepsilon k \rfloor$ candidates from $C'$ and all the candidates from $C_1\cup \cdots \cup C_{k - \lfloor \varepsilon k \rfloor}$ to the winning committee. Thus, the fraction of candidates in the winning committee from $C'$ is at most $\frac{\lfloor \varepsilon k \rfloor}{k} \leq \frac{\varepsilon k}{k} = \varepsilon$. This completes the proof.
\end{proof}

An interesting implication is that there exists no (sequential) $\lambda$-Thiele method that would satisfy IUAC and that would implement the idea of regressive proportionality. Such rules exist within the class of $\beta$-Phragm\'{e}n's rules. This shows that the rules from the classes of $\alpha/\beta$-Phragm\'{e}n's and $\lambda$-Thiele methods exhibit different properties.


\section{Degressive and Regressive Proportionality in the Euclidean Model}

We will now analyze degressive and regressive proportionality through experiments. Our goal is to understand how voters' satisfaction depends on using committee election rules implementing different types of proportionality.

\subsection{Distributions of Voters' Preferences}
We consider the one-dimensional Euclidean model, where each individual (a voter or a candidate) is represented as a point in the interval $[-1; 1]$. Intuitively, this point represents the position of the individual in the left-right political spectrum. The Euclidean model is commonly used in political science~\cite{DavisHinich66,Plott1967,enelow1984spatial,enelow1990advances,mckelvey1990,merrill1999unified,schofield2007spatial}, and---more recently---in computational social choice~\cite{ijcai/ElkindL15-dichpref,elk-fal-las-sko-sli-tal:c:2d-multiwinner,fal-szu-tal:utopia,fal-tal:between-prop-and-diversity}.

We draw the positions of individuals independently at random from a beta distribution, scaled into $[-1, 1]$. We consider four distributions: $\texttt{Beta}(\nicefrac{1}{2}, \nicefrac{1}{2})$, $\texttt{Beta}(\nicefrac{1}{2}, 2)$, $\texttt{Beta}(2, 2)$, and $\texttt{Beta}(2, 4)$.
The approval preferences of the voters are constructed from their positions as follows. We fix the approval radius $\xi \in \{0.1, 0.2, 0.3, 0.4, 0.5\}$, and assume that a voter $v$ approves a candidate $c$ if and only if $|v-c|\leq \xi$. We set a threshold of 0.5 for the approval radius---higher values would imply that there might exist voters that approve e.g. extreme right-wing candidates and left-wing candidates. The elections drawn this way belong to the candidate-interval domain~\cite{ijcai/ElkindL15-dichpref}. Further, each election instance from the candidate-interval domain can be obtained from the one-dimensional Euclidean model. 

\subsection{Voting Rules used in Simulations}

In our simulations we select winning committees using one of the following rules: 
\begin{enumerate}[\hspace{0pt}1.]
    \item Degressive: $\alpha$-Phragm\'{e}n's with $\alpha(i)=(\nicefrac{1}{2})^i$;
    \item Regressive: $\beta$-Phragm\'{e}n's with $\beta(x)=(\nicefrac{9}{10})^{100x}$;
    \item Linear: $\alpha$-Phragm\'{e}n's with $\alpha(i)=1$ (which equivalent to Phragm\'{e}n's Sequential rule).
\end{enumerate}

\subsection{Two Measures of Voters' Satisfaction}
We quantify the satisfaction of the voters from the elected committees using two measures:

\paragraph{Number of Representatives.} In this case the satisfaction of a voter $v$ from a committee $W$ equals to the number of committee members that $v$ approves, $\sat(v)=|W\cap A(v)|$.

\paragraph{Number of Satisfying Decisions.} Our second measure is based on the analysis of the voting committee model ~\cite{skow:multiwinner-models,ijcai2020-28}. The high level idea is the following. We assume that the voters and the candidates have preferences over a number of binary issues. The positions of the individuals are correlated with their preferences over the issues. Thus the voters and the candidates who are closer in the Euclidean space are more likely to have similar opinions regarding the issues. Next, we assume that the elected committee uses majoritarian voting to decide about the issues. We measure the satisfaction of a voter $v$ as the fraction of the committee's decisions consistent with $v$'s preferences. Formally, we use the following procedure.
    \begin{enumerate}[\hspace{0pt}1.]
        \item We generate $p$ issues according to the same beta distribution as the one from which we have sampled the voters and candidates---each issue is represented as a value from $[-1;1]$ which indicates the ideological characteristic of the issue. In our experiments we have used $p = 100$.
        \item We assign to each individual a $p$ dimensional binary vector, where in the $i$-th position of the vector we set 1 if the individual is \emph{for} the issue, and 0 if she is \emph{against}. In order to generate preferences over the issues we use the Bernoulli distribution, where the probability of an individual $\eta$ getting 1 in the $i$-th position in the vector depends on the position $x$ of an issue in the Euclidean space, and is given by the following formula:
            \begin{align*}
                p_{\eta}(x) = 
                \begin{cases}
                \frac{1}{\tau(1 - |\eta|)|\eta - x| + 1}, & \mbox{if } |x| > |\eta| \text { and } x\eta > 0,\\ 
                1, & \mbox{if } |x| < |\eta| \text { and } x\eta > 0,\\ 
                \frac{1}{(\delta|\eta| + \tau)|x| + 1}, & \mbox{if } x \eta \leq 0.
                \end{cases}
            \end{align*}
            
        The function $p_{\eta}$ for different values of $\eta$ (and $\tau=30$, $\delta=120$) is depicted in \Cref{fig:prob_of_1} in the appendix. Let us explain the form of this function through an example. Consider a center-left individual $\eta = -0.3$. We assume that the issue in the center $x = 0$ represents the status quo---accepting this issue does not change the state of the world, hence no voter opposes to it. E.g., the position $\eta = -0.3$ might correspond to the preferred tax rate at the level of 35\%, while the current tax rate---corresponding to position $x = 0$---is 31\%.
        Further:
        \begin{enumerate}[\hspace{0pt}a.]
            \item Every left-oriented issue $x$ that is closer to the center than the individual (i.e., $x\in [\eta, 0)$) is always approved by $\eta$. Accepting such an issue changes the status-quo towards the state that is preferred by the individual. 
            \item The more far-left the issue, the less likely it is that the individual accepts the issue. Our centre-left individual $\eta$ has an aversion to radicalism, thus for $x < \eta$ the probability function is increasing and convex.
            \item For right-oriented issues the function is decreasing and convex; the slope is greater than in case of more left-oriented issues, as the individual $\eta$ is centre-left. 
            \item The more radical the individual, the less probable it is that she accepts the issue with the opposite characteristics. For instance, for $\eta_1 > \eta_2 > 0$ and $x < 0$ we have that $p_{\eta_1}(x) < p_{\eta_2}(x)$.
        \end{enumerate}
        Such functions ensure that candidates that are closer to a voter $v$ are more likely to have similar preferences as $v$ regarding the issues (see \Cref{fig:coincided_pref} in the appendix). 
    \end{enumerate}
    Once we build an election with preferences over issues, we measure the satisfaction of each voter $v$ as the fraction of issues for which $v$'s preferences coincide with the committee's decision; recall that in the voting committee model the winning committee $W$ makes majoritarian decisions.

\begin{table}[t]
\centering
\begin{tabular}{ c|c|c|c|c|c} 
\texttt{Beta} & rule & \multicolumn{2}{c|}{\# representatives} & \multicolumn{2}{c}{\# sat. decisions}\\
params & & avg & std & avg & std \\
\toprule
        & degr & 5.753 & 0.923 & 0.683 & 0.121 \\ 
$(2,2)$ & lin & 6.548 & 2.254 & $\mathbf{0.686}$ & 0.125 \\ 
        & regr & $\mathbf{7.708}$ & 6.853 & 0.681 & 0.153\\ 
\midrule
        & degr & 6.977 & 1.201 & 0.681 & 0.128 \\ 
$(2,4)$ & lin & 8.590 & 3.312 & $\mathbf{0.687}$ & 0.140 \\ 
        & regr & $\mathbf{10.31}$ & 8.996 & 0.677 & 0.173 \\ 
\midrule
                      & degr & 6.990 & 1.543 & 0.555 & 0.233 \\ 
$(\nicefrac{1}{2},2)$ & lin & 11.16 & 6.120 & $\mathbf{0.659}$ & 0.154 \\
                      & regr & $\mathbf{14.45}$ & 11.74 & 0.651 & 0.321 \\ 
\midrule
                                    & degr & 5.137 & 0.847 & 0.667 & 0.153 \\ 
$(\nicefrac{1}{2},\nicefrac{1}{2})$ & lin & 5.674 & 2.101 & $\mathbf{0.668}$ & 0.155 \\ 
                                    & regr & $\mathbf{6.763}$ & 6.330 & 0.584 & 0.278 \\ 
\bottomrule
\end{tabular}
\caption{The total satisfaction of the voters for $\xi=0.2$.}\label{tab:numerical_results}
\vspace{-0.1cm}
\end{table}

\subsection{Results of the Simulations}

\begin{figure*}[ht]
\minipage{0.245\textwidth}
  \includegraphics[width=\linewidth]{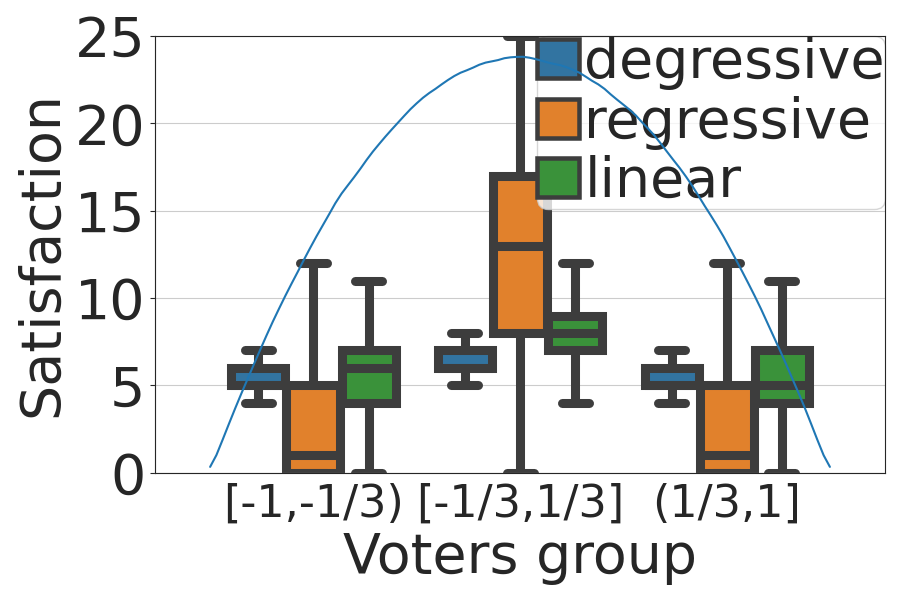}
  (a) \texttt{Beta}$(2,2)$
\endminipage
\minipage{0.245\textwidth}
  \includegraphics[width=\linewidth]{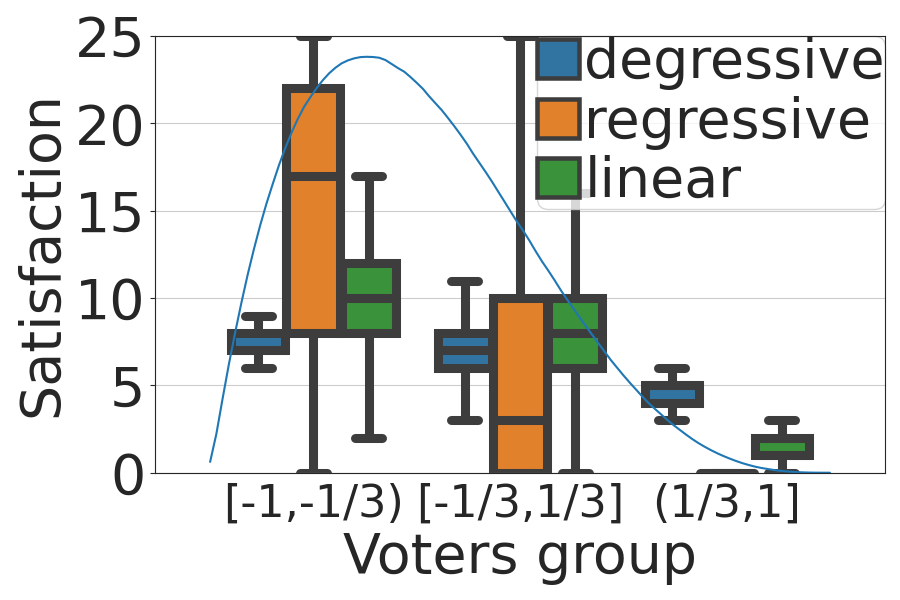}
  (b) \texttt{Beta}$(2,4)$
\endminipage
\minipage{0.245\textwidth}
  \includegraphics[width=\linewidth]{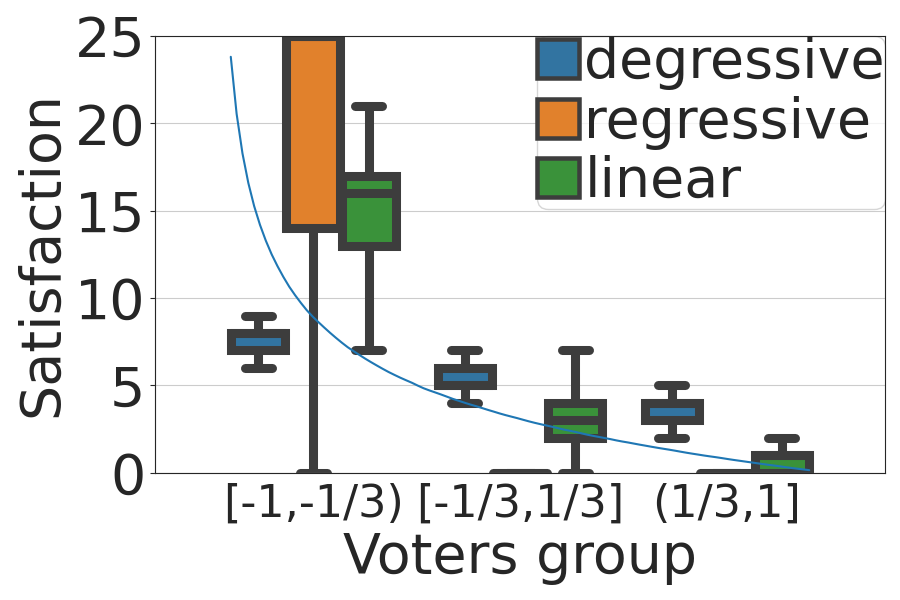}
  (c) \texttt{Beta}$(\nicefrac{1}{2},2)$
\endminipage
\minipage{0.245\textwidth}
  \includegraphics[width=\linewidth]{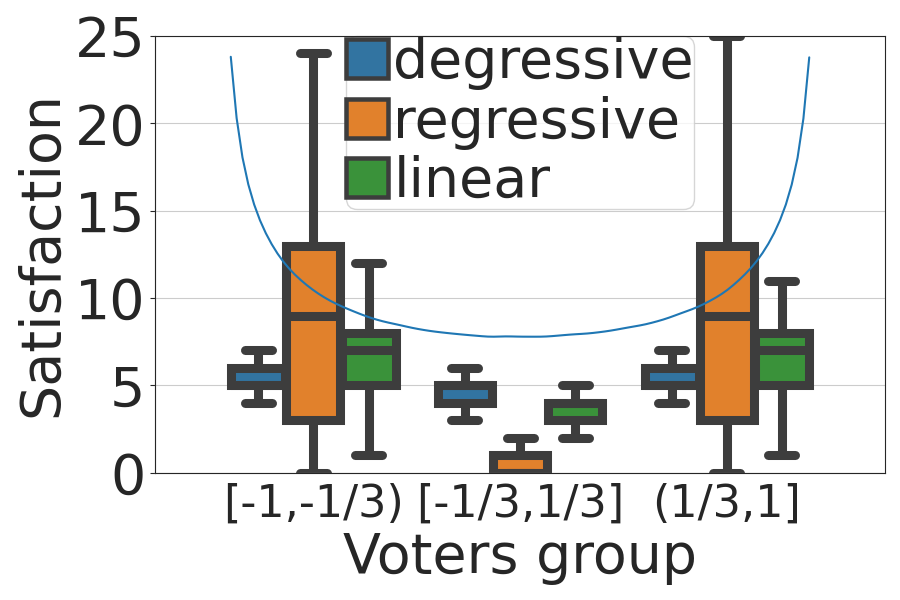}
  (d) \texttt{Beta}$(\nicefrac{1}{2}, \nicefrac{1}{2})$
\endminipage
\caption{Box plots with the distribution of voters' satisfaction (measured as the number of representatives) for different society models (beta  distributions). Acceptance radius is $\xi = 0.2$. In each plot the blue line depicts the density of the distribution from which we sampled voters and candidates.}\label{fig:basic_satisfaction}
\end{figure*}

\begin{figure*}[ht]
\minipage{0.245\textwidth}
  \includegraphics[width=\linewidth]{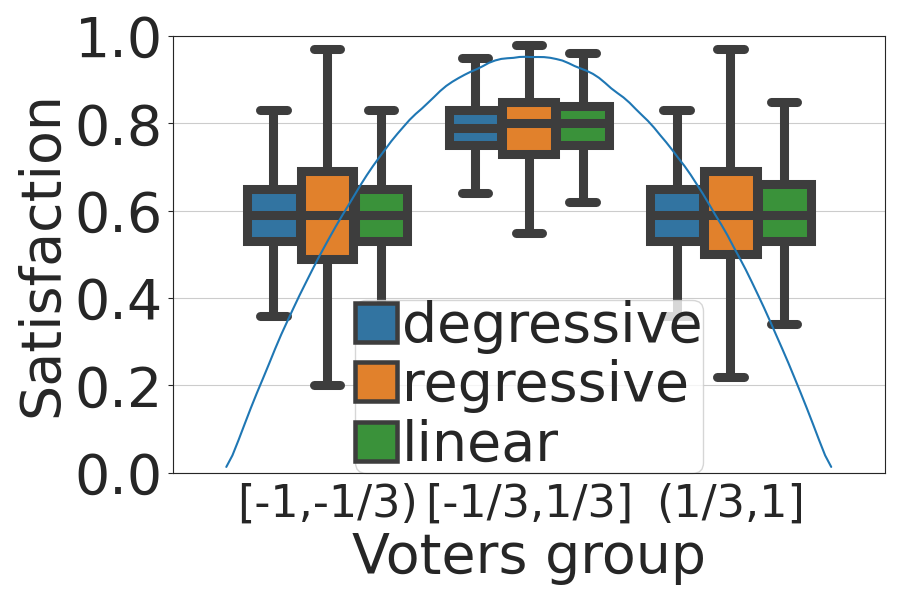}
  (a) \texttt{Beta}$(2,2)$
\endminipage
\minipage{0.245\textwidth}
  \includegraphics[width=\linewidth]{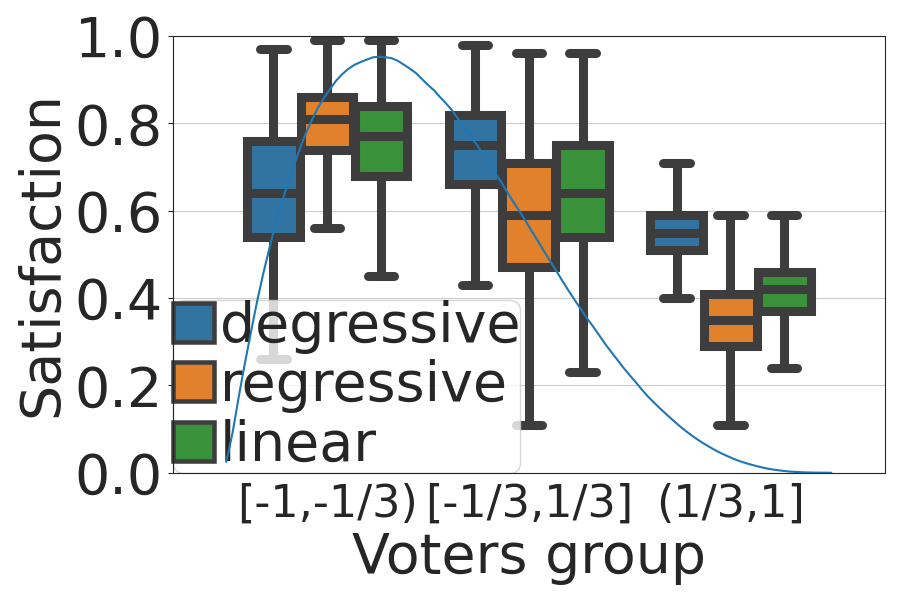}
  (b) \texttt{Beta}$(2,4)$
\endminipage
\minipage{0.245\textwidth}
  \includegraphics[width=\linewidth]{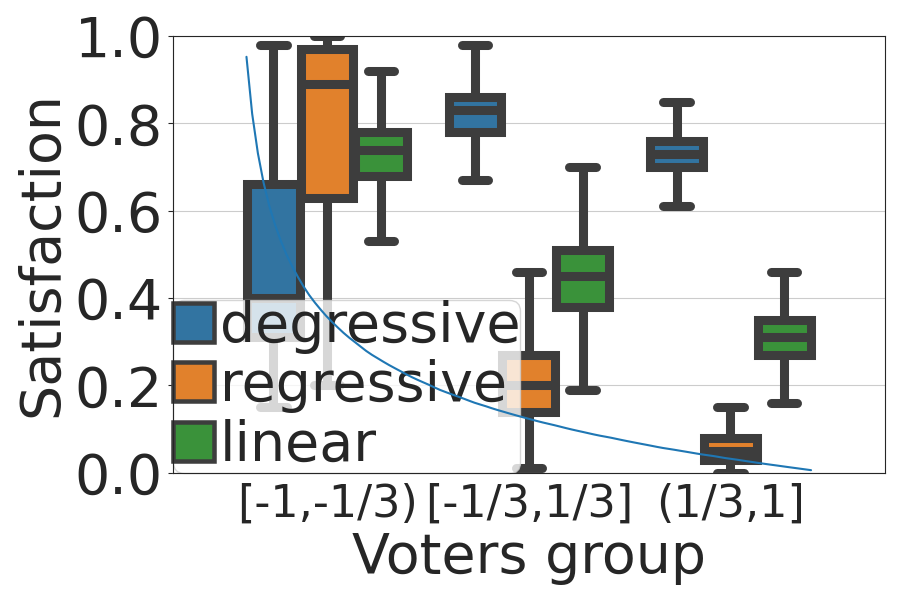}
  (c) \texttt{Beta}$(\nicefrac{1}{2},2)$
\endminipage
\minipage{0.245\textwidth}
  \includegraphics[width=\linewidth]{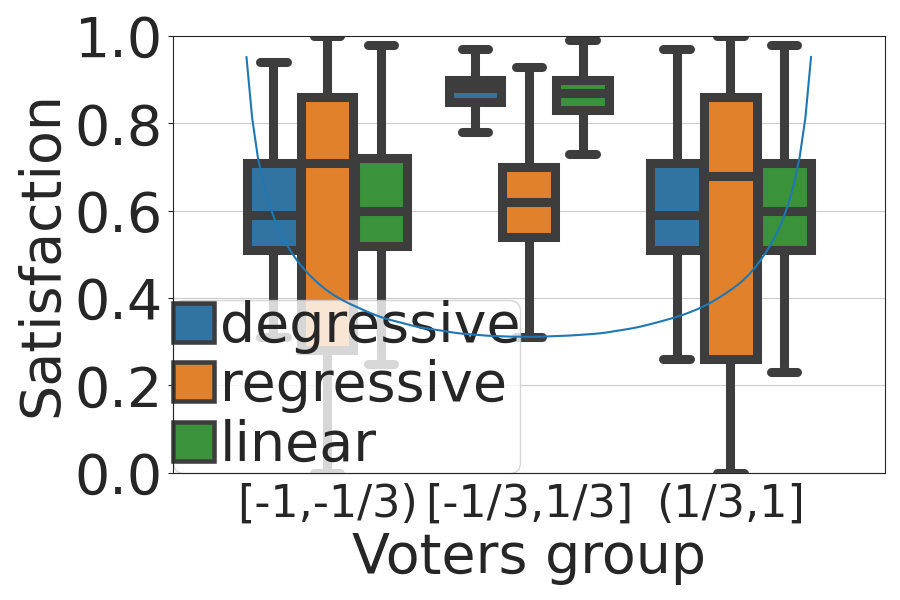}
  (d) \texttt{Beta}$(\nicefrac{1}{2}, \nicefrac{1}{2})$
\endminipage

\caption{The distribution of voters' satisfaction in the voting committee model for different beta  distributions. Acceptance radius is $\xi = 0.2$. In each plot the blue line depicts the density of the distribution from which we sampled voters and candidates.}\label{fig:issues_satisfaction}
\end{figure*}

We present the results of the aggregated satisfaction of the voters in the form of box plots in \Cref{fig:basic_satisfaction} and \Cref{fig:issues_satisfaction}. In order to simplify the visual presentation we divided the voters into three groups based on their position on the Euclidean line: $[-1, -\nicefrac{1}{3}), [-\nicefrac{1}{3}, \nicefrac{1}{3}], (\nicefrac{1}{3}, 1]$. In the simulations we consider instances with $n=200$ voters, $m=150$ candidates and for the committee size $k=25$. What is more, we set the acceptance radius $\tau =0.2$ and the parameters of the probability function $p_{\eta}$ to: $\tau = 30$, $\delta = 120$. We also checked several others sets of parameters (e.g. \{$\tau = 5$, $\delta = 20$\}, \{$\tau = 10$, $\delta = 60$\}), but we found that  the key observations and regularities stay the same. We ran 1000 simulations for each scenario. Additional results of the simulations can be found in the appendix.

Additionally in \Cref{tab:numerical_results} we give numerical values quantifying the total satisfaction of the voters.

From the experiments we draw the following conclusions:
\begin{enumerate}[\hspace{0pt}1.]
    \item Except for the case of polarized societies (i.e., scenarios (c) and (d) in Figures~\ref{fig:basic_satisfaction} and \ref{fig:issues_satisfaction}), we observe a positive correlation between the voters' satisfactions quantified according to our two measures. This suggests that for such societies  voters' satisfactions from the decisions made by the committee is related to the number of representatives these voters get in the elected committee. On the other hand, for the polarized societies there is no such a correlation. For example, in scenario (d) the groups of voters from less populous (and underrepresented) areas are more happy with the decisions made by the committees than the voters from the populous well represented poles. In such cases regressive proportional rules result in the distributions of the voters' satisfactions that more closely resemble densities of the voters' distributions.  
    \item We observe that for regressive proportional rules the shapes of the distributions of the voters' satisfactions (measured in either of the two ways) reflect the shapes of the densities of the voters' distributions. Interestingly, this relation is reflected to a slightly smaller extent for linear-proportional rules, and is generally not observed for the rules following the principle of degressive proportionality (see the plots for (c) and (d) in \Cref{fig:issues_satisfaction}). In our opinion this weakens the arguments in favor of degressive proportionality that are sometimes raised in the literature.
    \item In the voting committee model degressive-proportional rules favor less densely populated areas compared to the other rules, which is especially visible in case of asymmetric voters' distributions (cf. (b) and (c) in \Cref{fig:issues_satisfaction}).
    \item The largest variance is observed for regressive-proportional rules---specifically for the case when the voters' satisfaction is measured as the number of representatives, for \texttt{Beta}$(\nicefrac{1}{2}, 2)$. What is more, in case of \texttt{Beta}$(2, 4)$, around half of the voters from $[-1, -\nicefrac{1}{3})$ have the satisfaction higher than 17 and around 25\% have the satisfaction lower than 8.
    \item In the voting committee model the highest average satisfaction of the voters from the committees' decisions is observed for the rules that follow linear proportionality. If we measure voters' satisfaction as the number of their representatives in the elected committees, then the highest total satisfaction is attained by rules that follow regressive proportionality (consult \Cref{tab:numerical_results}).   
\end{enumerate}

\section{Conclusion}

We have defined a family of committee election rules that extend Phragm\'{e}n's Sequential Rule. These rules span the spectrum of different types of proportionality. We have assessed the worst-case guarantees that these rules provide to groups of voters with similar preferences, and analyzed how these rules treat voters assuming the voters and the candidates are represented as points in the one-dimensional Euclidean space. 

\bibliographystyle{abbrvnat}
\bibliography{aaai22}

\newpage
\appendix

\section{Figures Not Included in the Main Text}
In this section we present figures that were not included in the main text.




\begin{figure}[ht]
\begin{center}

\minipage{0.75\textwidth}
  \includegraphics[width=\linewidth]{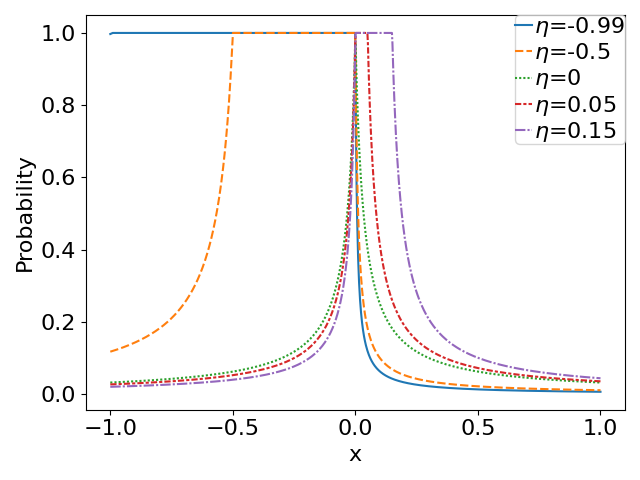}
\endminipage\hfill

\end{center}
\caption{Probability of an individual $\eta$ being \emph{for} an issue~$x$ in the voting committee model ($\tau=30$, $\delta=120$).}\label{fig:prob_of_1}
\end{figure}

\begin{figure}[h]
\begin{center}

\minipage{0.45\textwidth}
  \includegraphics[width=\linewidth]{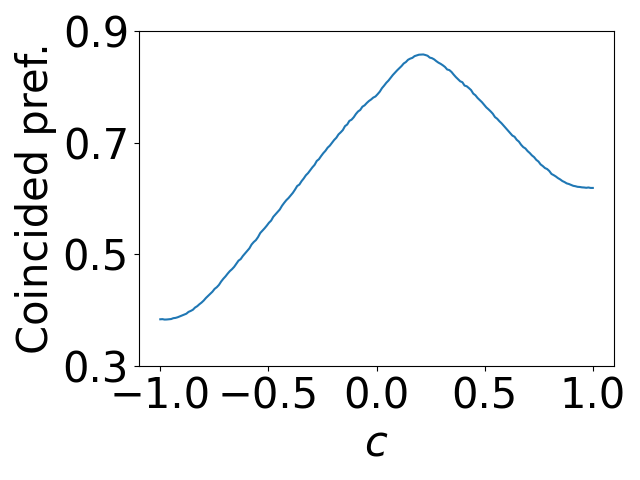}
  (a) voter $v = 0.2$
\endminipage\hfill
\minipage{0.45\textwidth}
  \includegraphics[width=\linewidth]{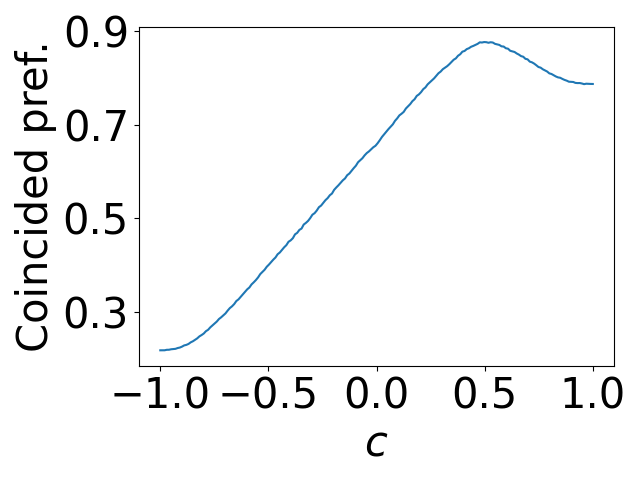}
  (b) voter $v = 0.5$
\endminipage

\end{center}
\caption{The probability that a candidate $c$ has the same preference as a voter $v$ regarding a randomly selected issue.}\label{fig:coincided_pref}
\end{figure}

\newpage
\section{Results of the Simulations Not Included in the Main Text}
In \Cref{fig:basic_satisfaction_0.1_appendix,fig:issues_satisfaction_0.1_appendix,fig:basic_satisfaction_0.3_appendix,fig:issues_satisfaction_0.3_appendix,fig:basic_satisfaction_0.4_appendix,fig:issues_satisfaction_0.4_appendix,fig:basic_satisfaction_0.5_appendix,fig:issues_satisfaction_0.5_appendix} and \Cref{tab:numerical_results_0.1,tab:numerical_results_0.3,tab:numerical_results_0.4,tab:numerical_results_0.5} we present the results of the simulations that were not included in the main text. The aggregated statistics for $\xi \in \{0.1, 0.3, 0.4, 0.5\}$ follow the same patterns as for $\xi =0.2$ (the other parameters remain the same). There are some minor differences in case of $\xi = \{0.3, 0.4, 0.5\}$---if we measure the satisfaction of the voters as their satisfaction from the committees' decisions, then for \texttt{Beta}$(2,4)$ and \texttt{Beta}$(\nicefrac{1}{2},2)$ the rules that follow regressive proportionality attain slightly higher average satisfaction than the linear proportional rules (see \Cref{tab:numerical_results_0.3,tab:numerical_results_0.4,tab:numerical_results_0.5}). What is more, in case of $\xi=0.5$ for \texttt{Beta}$(2,2)$ and \texttt{Beta}$(\nicefrac{1}{2}, \nicefrac{1}{2})$ the rules that follow regressive and degressive proportionality respectively, again attain slightly higher satisfaction.

\begin{table}[t]
\centering
\begin{tabular}{ c|c|c|c|c|c} 
\texttt{Beta} & rule & \multicolumn{2}{c|}{\# representatives} & \multicolumn{2}{c}{\# sat. decisions}\\
params & & avg & std & avg & std \\
\toprule
        & degr & 3.125 & 0.781 & 0.685 & 0.122 \\ 
$(2,2)$ & lin & 3.473 & 1.419 & $\mathbf{0.686}$ & 0.126 \\ 
        & regr & $\mathbf{4.069}$ & 3.736 & 0.684 & 0.144\\ 
\midrule
        & degr & 3.875 & 0.951 & 0.688 & 0.123 \\ 
$(2,4)$ & lin & 4.566 & 1.997 & $\mathbf{0.693}$ & 0.141 \\ 
        & regr & $\mathbf{5.542}$ & 5.453 & 0.686 & 0.175 \\ 
\midrule
                      & degr & 4.379 & 1.438 & 0.593 & 0.188 \\ 
$(\nicefrac{1}{2},2)$ & lin & 7.111 & 4.988 & $\mathbf{0.657}$ & 0.251 \\
                      & regr & $\mathbf{10.27}$ & 11.56 & 0.649 & 0.332 \\ 
\midrule
                                    & degr & 2.918 & 0.857 & 0.668 & 0.155 \\ 
$(\nicefrac{1}{2},\nicefrac{1}{2})$ & lin & 3.398 & 1.867 & $\mathbf{0.669}$ & 0.156 \\ 
                                    & regr & $\mathbf{4.501}$ & 5.491 & 0.559 & 0.303 \\ 
\bottomrule
\end{tabular}
\caption{The total satisfaction of the voters assessed through experiments for $\xi=0.1$.}\label{tab:numerical_results_0.1}
\end{table}

\begin{table}[t]
\centering
\begin{tabular}{ c|c|c|c|c|c} 
\texttt{Beta} & rule & \multicolumn{2}{c|}{\# representatives} & \multicolumn{2}{c}{\# sat. decisions}\\
params & & avg & std & avg & std \\
\toprule
        & degr & 8.353 & 1.030 & 0.686 & 0.124 \\ 
$(2,2)$ & lin & 9.624 & 3.152 & $\mathbf{0.687}$ & 0.127 \\ 
        & regr & $\mathbf{11.18}$ & 8.972 & 0.684 & 0.148\\ 
\midrule
        & degr & 10.03 & 1.766 & 0.659 & 0.143 \\ 
$(2,4)$ & lin & 12.66 & 4.820 & 0.690 & 0.149 \\ 
        & regr & $\mathbf{14.71}$ & 10.27 & $\mathbf{0.691}$ & 0.159 \\ 
\midrule
                      & degr & 9.362 & 1.667 & 0.527 & 0.257 \\ 
$(\nicefrac{1}{2},2)$ & lin & 14.41 & 6.587 & 0.648 & 0.113 \\
                      & regr & $\mathbf{17.39}$ & 10.97 & $\mathbf{0.655}$ & 0.193 \\ 
\midrule
                                    & degr & 7.241 & 1.146 & 0.669 & 0.154 \\ 
$(\nicefrac{1}{2},\nicefrac{1}{2})$ & lin & 7.780 & 2.279 & $\mathbf{0.669}$ & 0.154 \\ 
                                    & regr & $\mathbf{8.664}$ & 6.451 & 0.617 & 0.238 \\ 
\bottomrule
\end{tabular}
\caption{The total satisfaction of the voters assessed through experiments for $\xi=0.3$.}\label{tab:numerical_results_0.3}
\end{table}

\begin{table}[t]
\centering
\begin{tabular}{ c|c|c|c|c|c} 
\texttt{Beta} & rule & \multicolumn{2}{c|}{\# representatives} & \multicolumn{2}{c}{\# sat. decisions}\\
params & & avg & std & avg & std \\
\toprule
        & degr & 11.12 & 1.873 & 0.677 & 0.114 \\ 
$(2,2)$ & lin & 12.63 & 3.968 & 0.683 & 0.125 \\ 
        & regr & $\mathbf{14.36}$ & 9.789 & $\mathbf{0.684}$ & 0.142\\ 
\midrule
        & degr & 12.89 & 1.995 & 0.677 & 0.122 \\ 
$(2,4)$ & lin & 16.52 & 6.134 & 0.694 & 0.142 \\ 
        & regr & $\mathbf{18.20}$ & 9.607 & $\mathbf{0.695}$ & 0.142 \\ 
\midrule
                      & degr & 11.49 & 1.975 & 0.549 & 0.235 \\ 
$(\nicefrac{1}{2},2)$ & lin & 17.04 & 6.609 & 0.633 & 0.131 \\
                      & regr & $\mathbf{19.50}$ & 9.828 & $\mathbf{0.648}$ & 0.115 \\ 
\midrule
                                    & degr & 9.243 & 1.545 & $\mathbf{0.669}$ & 0.154 \\ 
$(\nicefrac{1}{2},\nicefrac{1}{2})$ & lin & 9.701 & 2.445 & 0.668 & 0.157 \\ 
                                    & regr & $\mathbf{10.33}$ & 6.030 & 0.636 & 0.218 \\ 
\bottomrule
\end{tabular}
\caption{The total satisfaction of the voters assessed through experiments for $\xi=0.4$.}\label{tab:numerical_results_0.4}
\end{table}

\begin{table}[t]
\centering
\begin{tabular}{ c|c|c|c|c|c} 
\texttt{Beta} & rule & \multicolumn{2}{c|}{\# representatives} & \multicolumn{2}{c}{\# sat. decisions}\\
params & & avg & std & avg & std \\
\toprule
        & degr & 13.26 & 1.795 & 0.656 & 0.093 \\ 
$(2,2)$ & lin & 15.67 & 5.541 & 0.684 & 0.123 \\ 
        & regr & $\mathbf{17.25}$ & 9.561 & $\mathbf{0.685}$ & 0.137\\ 
\midrule
        & degr & 15.95 & 4.269 & 0.683 & 0.119 \\ 
$(2,4)$ & lin & 20.03 & 6.296 & 0.694 & 0.129 \\ 
        & regr & $\mathbf{21.11}$ & 7.888 & $\mathbf{0.695}$ & 0.132 \\ 
\midrule
                      & degr & 13.62 & 2.159 & 0.569 & 0.214 \\ 
$(\nicefrac{1}{2},2)$ & lin & 19.21 & 6.331 & 0.613 & 0.161 \\
                      & regr & $\mathbf{21.14}$ & 8.525 & $\mathbf{0.628}$ & 0.135 \\ 
\midrule
                                    & degr & 11.41 & 1.585 & $\mathbf{0.659}$ & 0.147 \\ 
$(\nicefrac{1}{2},\nicefrac{1}{2})$ & lin & 11.62 & 1.872 & 0.658 & 0.154 \\ 
                                    & regr & $\mathbf{11.88}$ & 5.106 & 0.650 & 0.197 \\ 
\bottomrule
\end{tabular}
\caption{The total satisfaction of the voters assessed through experiments for $\xi=0.5$.}\label{tab:numerical_results_0.5}
\end{table}

\begin{figure*}[!t]
\begin{center}
\minipage{0.245\textwidth}
  \includegraphics[width=\linewidth]{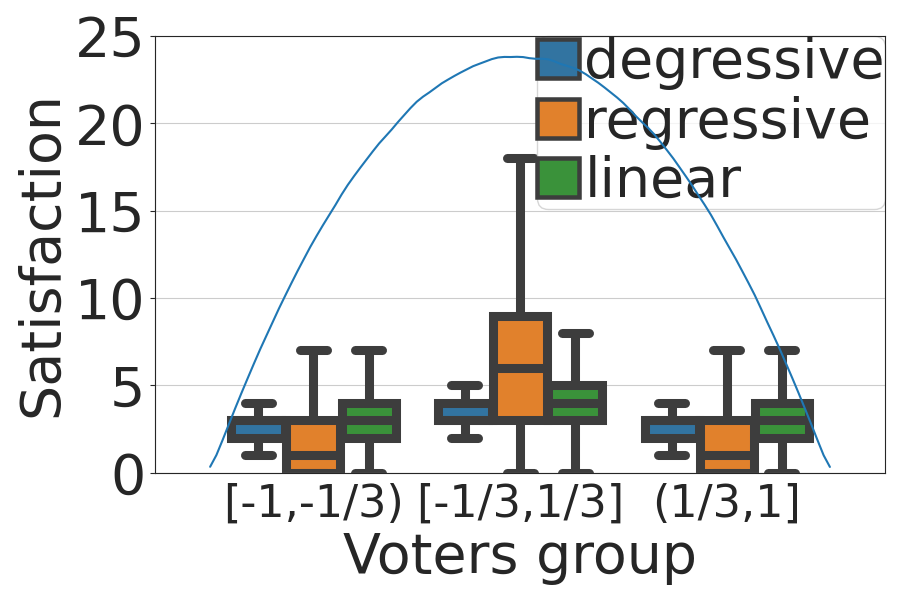}
  (a) \texttt{Beta}$(2,2)$
\endminipage
\minipage{0.245\textwidth}
  \includegraphics[width=\linewidth]{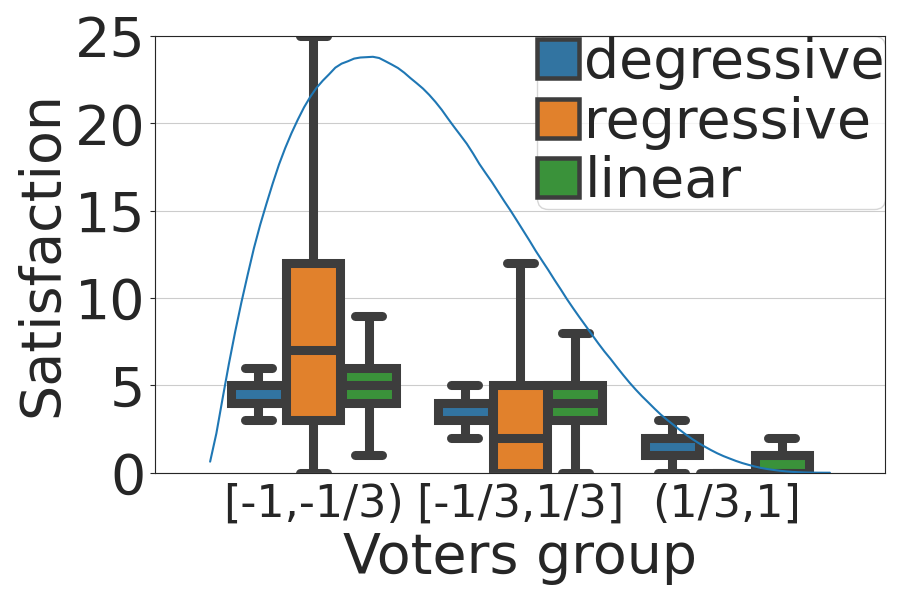}
  (b) \texttt{Beta}$(2,4)$
\endminipage
\minipage{0.245\textwidth}
  \includegraphics[width=\linewidth]{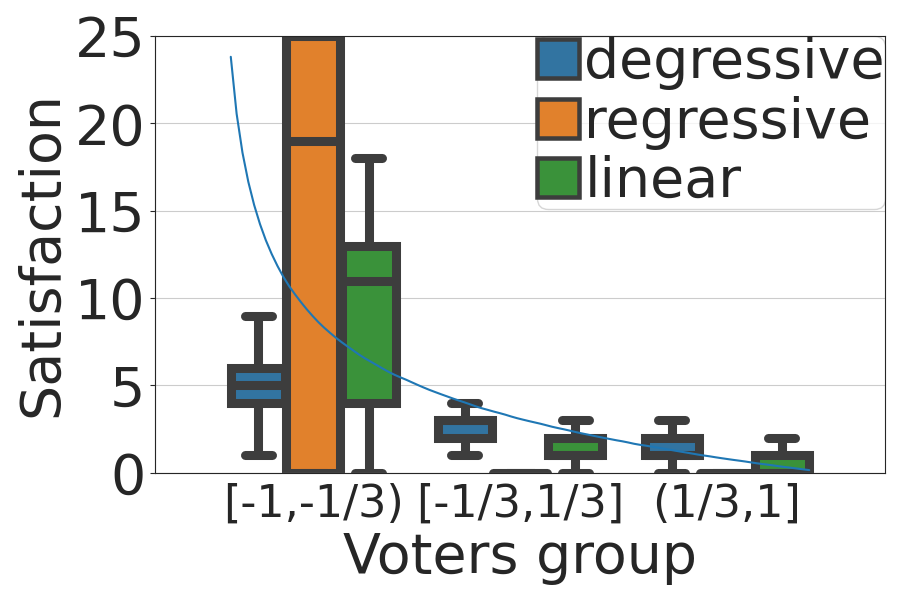}
  (c) \texttt{Beta}$(\nicefrac{1}{2},2)$
\endminipage
\minipage{0.245\textwidth}
  \includegraphics[width=\linewidth]{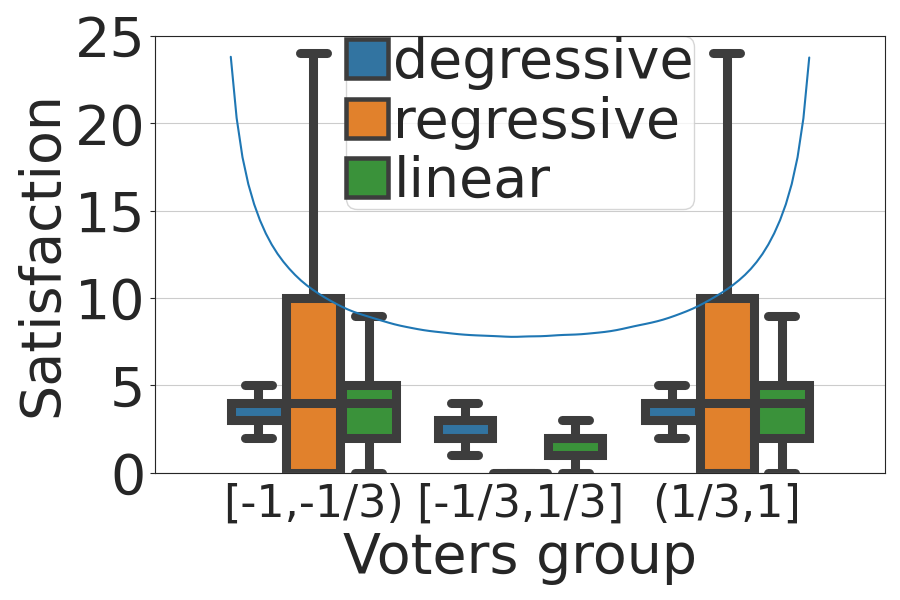}
  (d) \texttt{Beta}$(\nicefrac{1}{2}, \nicefrac{1}{2})$
\endminipage

\end{center}
\caption{Box plots with the distribution of voters' satisfaction for different society models (beta  distributions). Acceptance radius: $\xi = 0.1$. In each plot the blue line depicts the density of the distribution from which we sampled the voters and candidates.}\label{fig:basic_satisfaction_0.1_appendix}
\end{figure*}

\begin{figure*}[!t]
\begin{center}

\minipage{0.245\textwidth}
  \includegraphics[width=\linewidth]{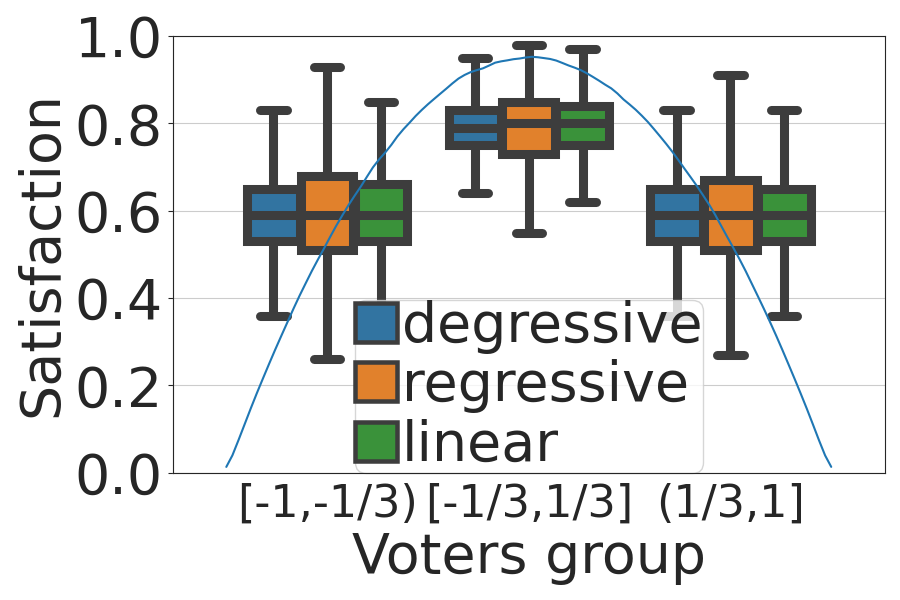}
  (a) \texttt{Beta}$(2,2)$
\endminipage
\minipage{0.245\textwidth}
  \includegraphics[width=\linewidth]{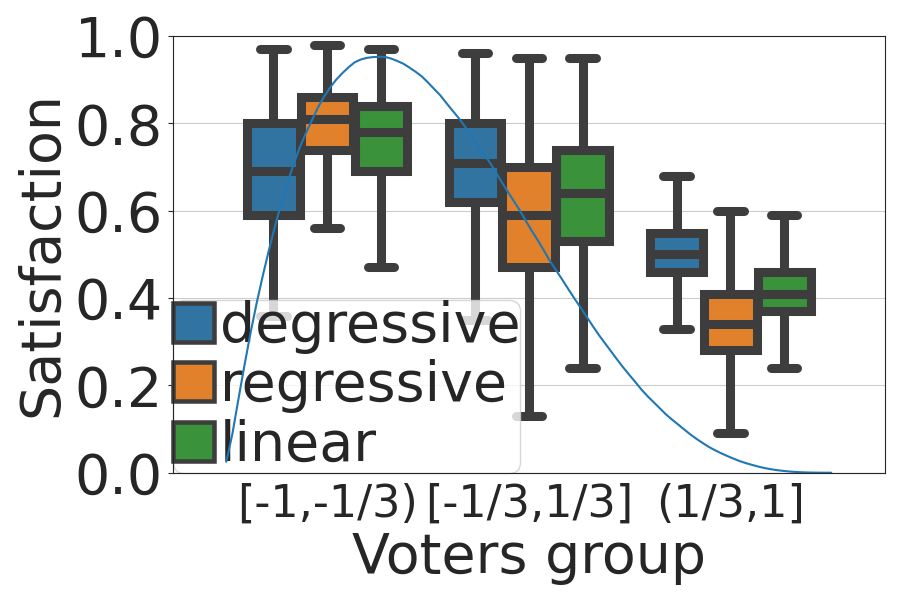}
  (b) \texttt{Beta}$(2,4)$
\endminipage
\minipage{0.245\textwidth}
  \includegraphics[width=\linewidth]{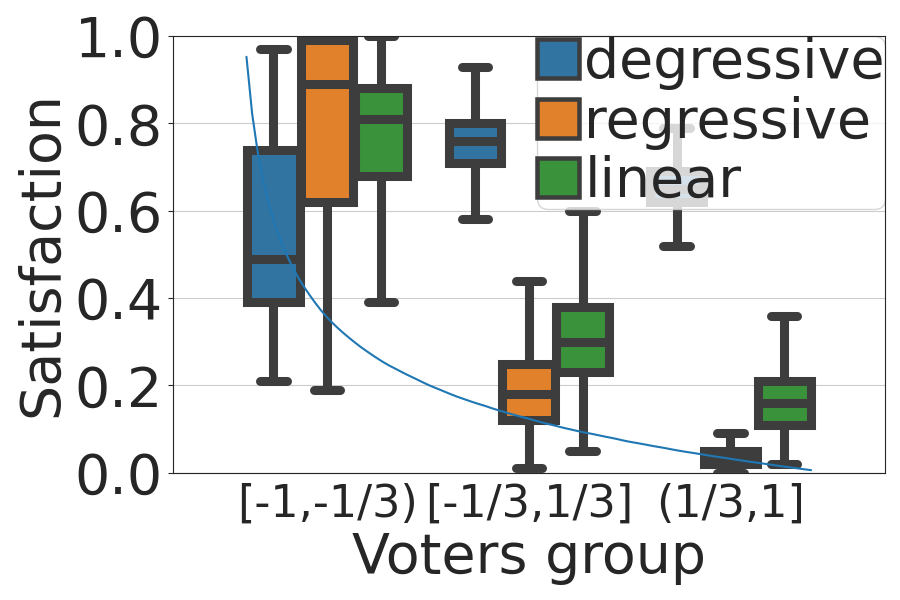}
  (c) \texttt{Beta}$(\nicefrac{1}{2},2)$
\endminipage
\minipage{0.245\textwidth}
  \includegraphics[width=\linewidth]{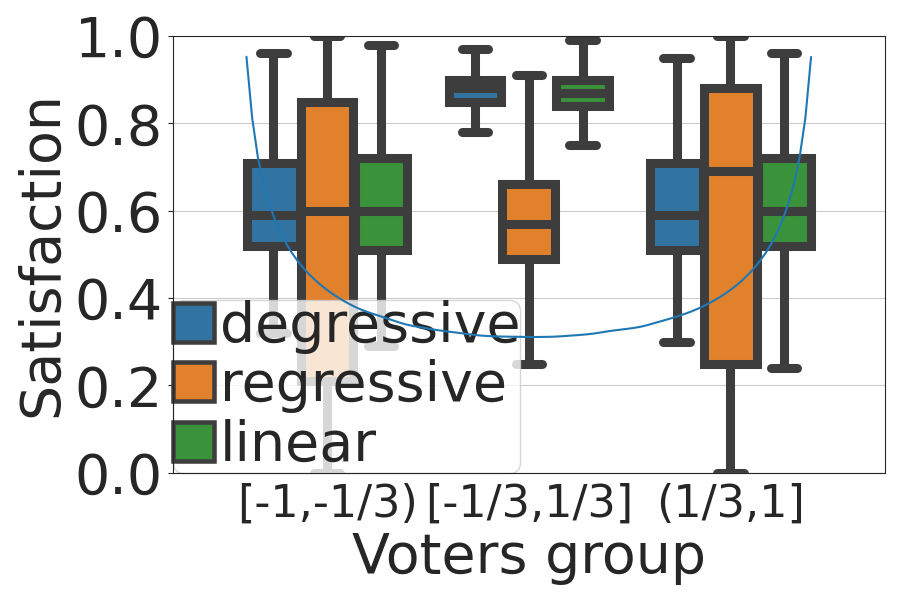}
  (d) \texttt{Beta}$(\nicefrac{1}{2}, \nicefrac{1}{2})$
\endminipage

\end{center}
\caption{Box plots with the distribution of voters' satisfaction in the Voting Committee Model for different society models (beta  distributions). Acceptance radius $\xi = 0.1$. In each plot the blue line depicts the density of the distribution from which we sampled voters and candidates.}\label{fig:issues_satisfaction_0.1_appendix}
\end{figure*}

\begin{figure*}[!t]
\begin{center}
\minipage{0.245\textwidth}
  \includegraphics[width=\linewidth]{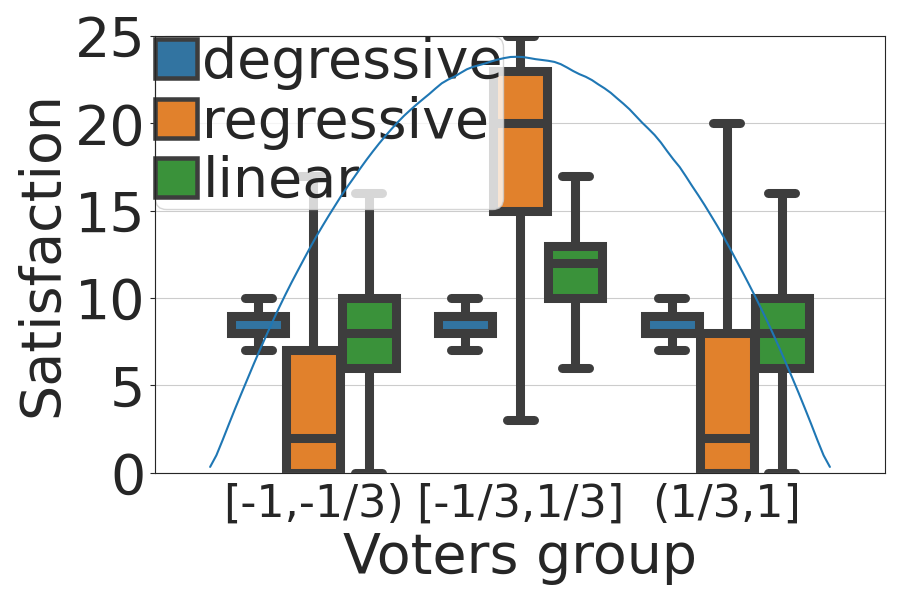}
  (a) \texttt{Beta}$(2,2)$
\endminipage
\minipage{0.245\textwidth}
  \includegraphics[width=\linewidth]{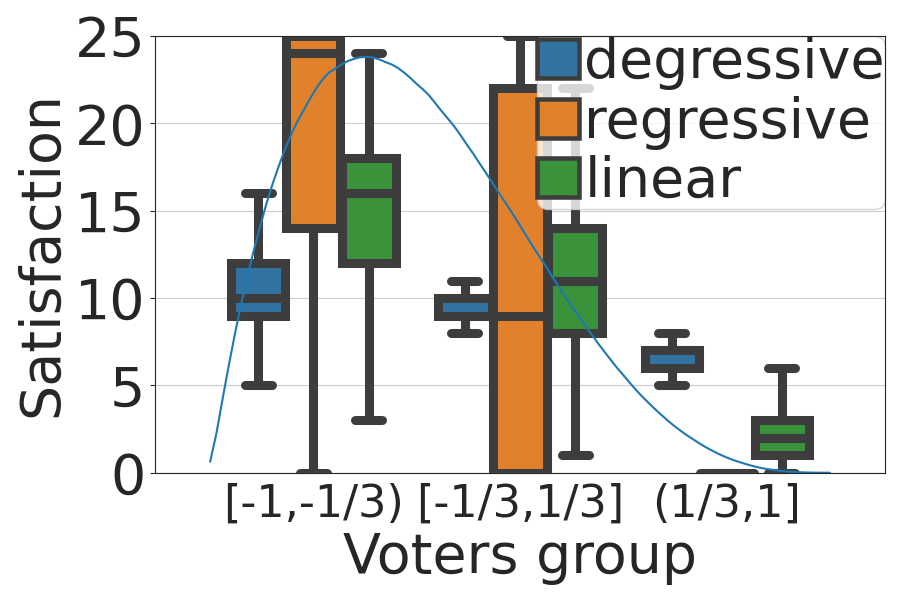}
  (b) \texttt{Beta}$(2,4)$
\endminipage
\minipage{0.245\textwidth}
  \includegraphics[width=\linewidth]{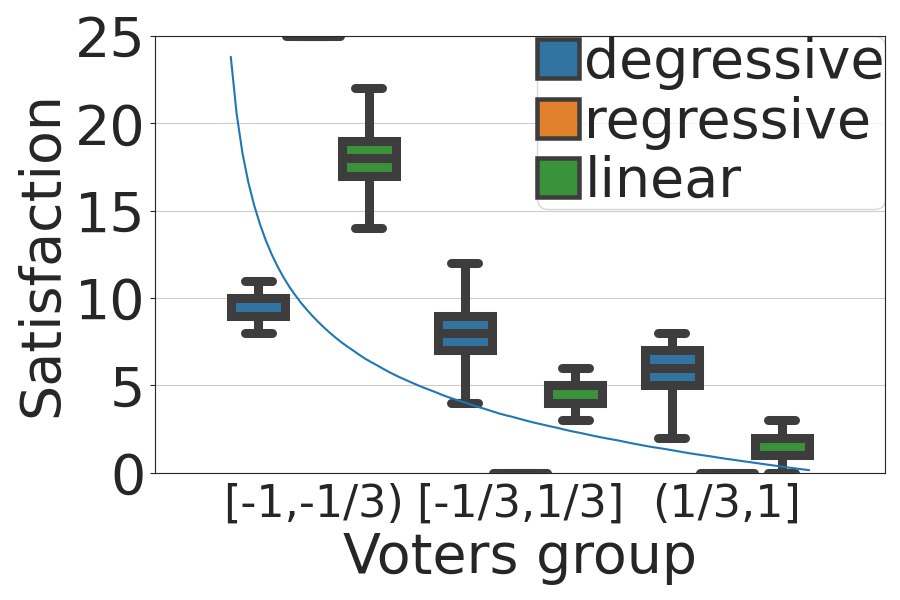}
  (c) \texttt{Beta}$(\nicefrac{1}{2},2)$
\endminipage
\minipage{0.245\textwidth}
  \includegraphics[width=\linewidth]{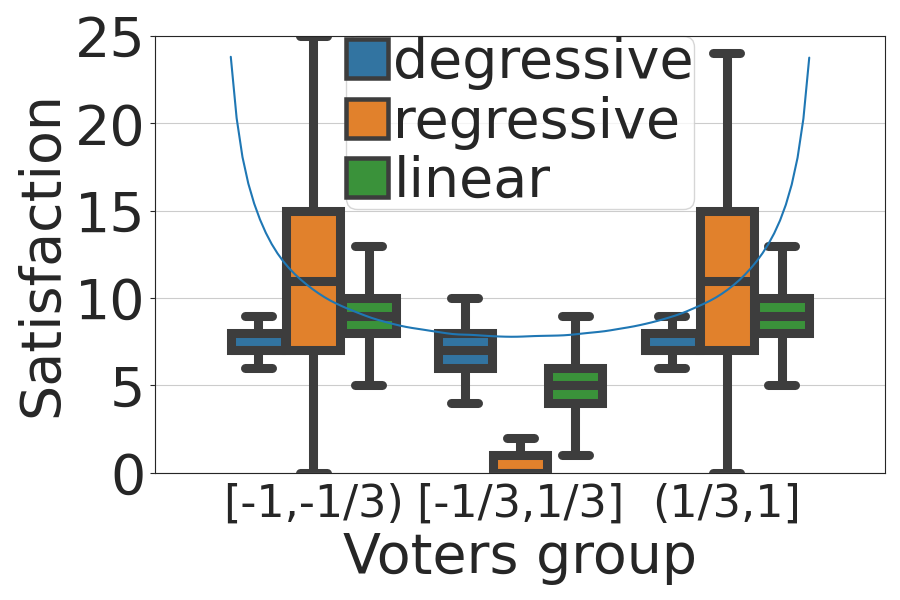}
  (d) \texttt{Beta}$(\nicefrac{1}{2}, \nicefrac{1}{2})$
\endminipage

\end{center}
\caption{Box plots with the distribution of voters' satisfaction for different society models (beta  distributions). Acceptance radius $\xi = 0.3$. In each plot the blue line depicts the density of the distribution from which we sampled voters and candidates.}\label{fig:basic_satisfaction_0.3_appendix}
\end{figure*}

\begin{figure*}[!t]
\begin{center}

\minipage{0.245\textwidth}
  \includegraphics[width=\linewidth]{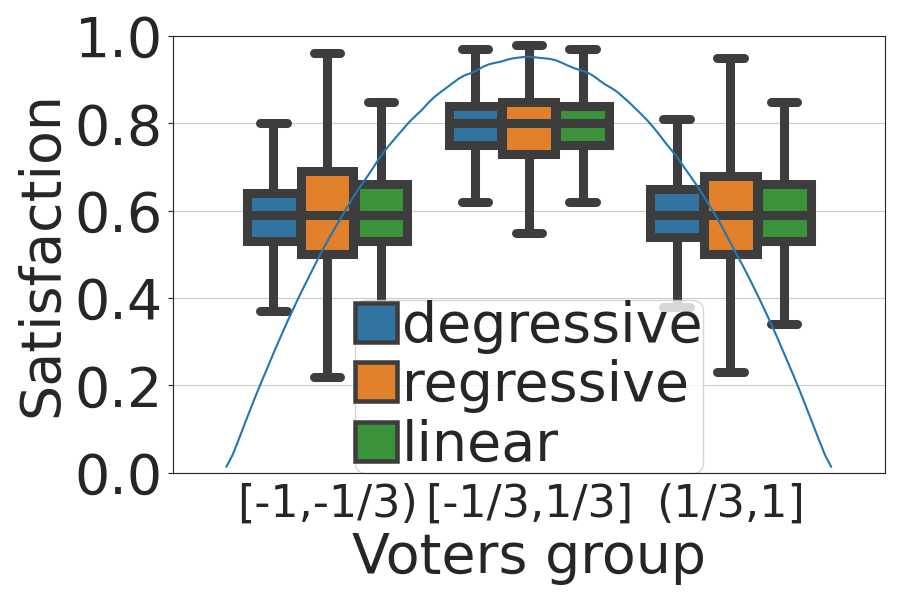}
  (a) \texttt{Beta}$(2,2)$
\endminipage
\minipage{0.245\textwidth}
  \includegraphics[width=\linewidth]{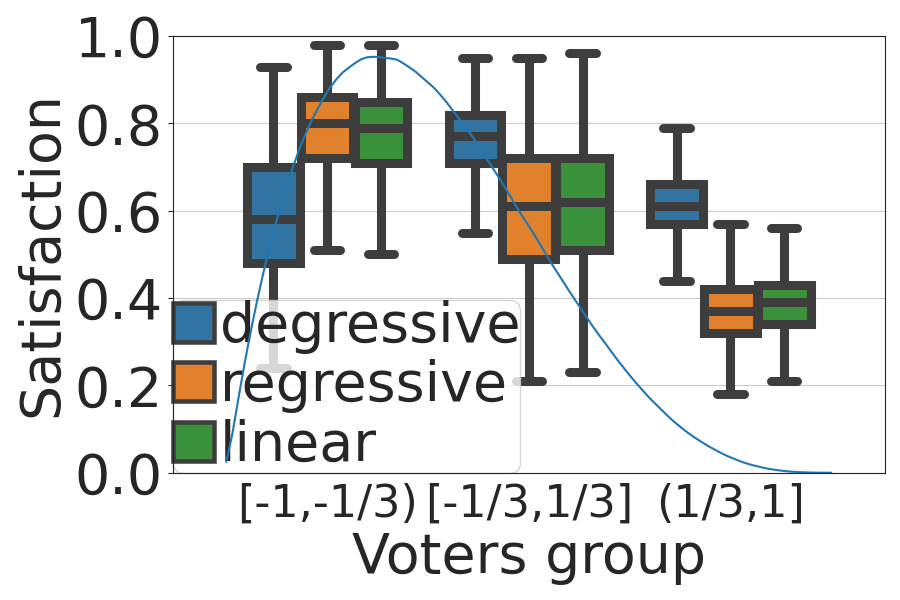}
  (b) \texttt{Beta}$(2,4)$
\endminipage
\minipage{0.245\textwidth}
  \includegraphics[width=\linewidth]{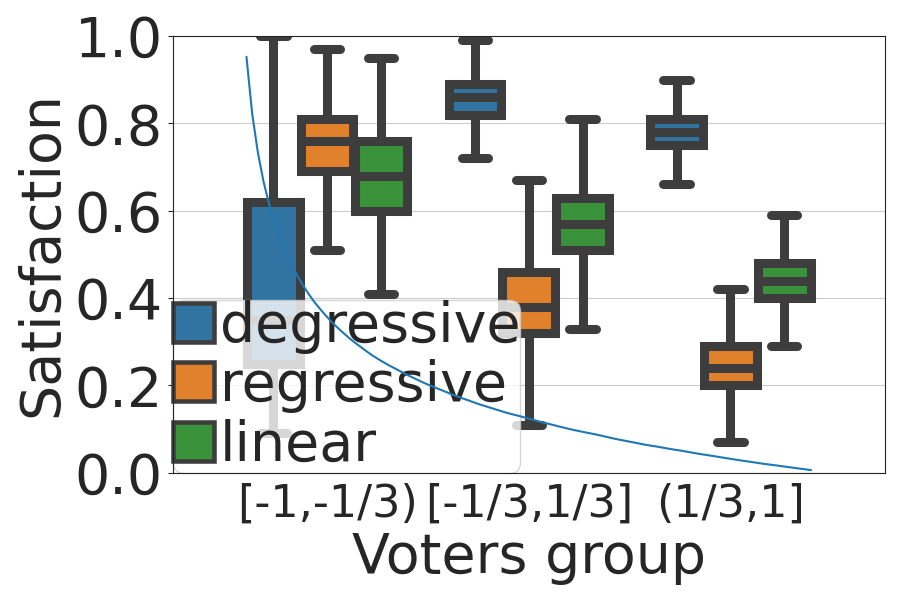}
  (c) \texttt{Beta}$(\nicefrac{1}{2},2)$
\endminipage
\minipage{0.245\textwidth}
  \includegraphics[width=\linewidth]{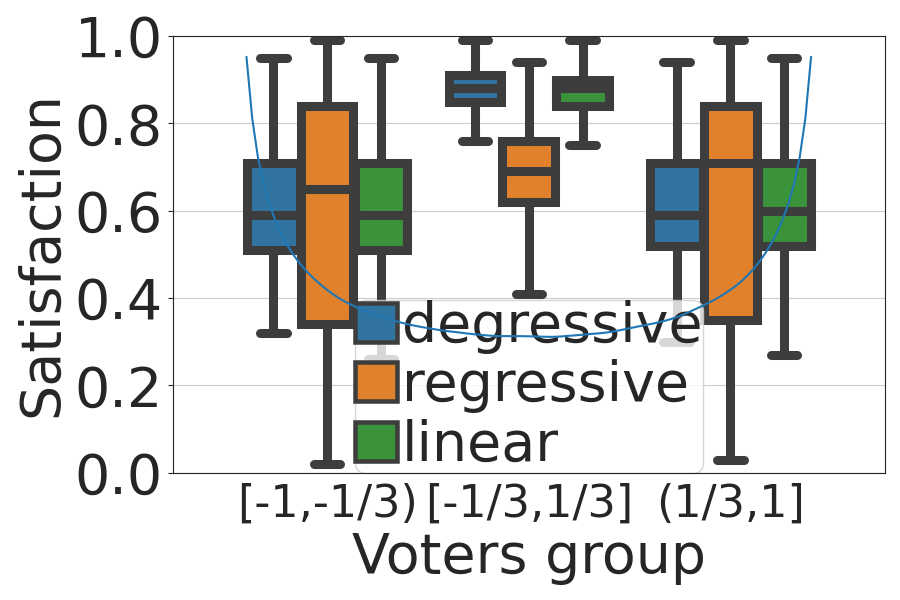}
  (d) \texttt{Beta}$(\nicefrac{1}{2}, \nicefrac{1}{2})$
\endminipage

\end{center}
\caption{Box plots with the distribution of voters' satisfaction in the Voting Committee Model for different society models (beta  distributions). Acceptance radius $\xi = 0.3$. In each plot the blue line depicts the density of the distribution from which we sampled voters and candidates.}\label{fig:issues_satisfaction_0.3_appendix}
\end{figure*}

\begin{figure*}[!t]
\begin{center}
\minipage{0.245\textwidth}
  \includegraphics[width=\linewidth]{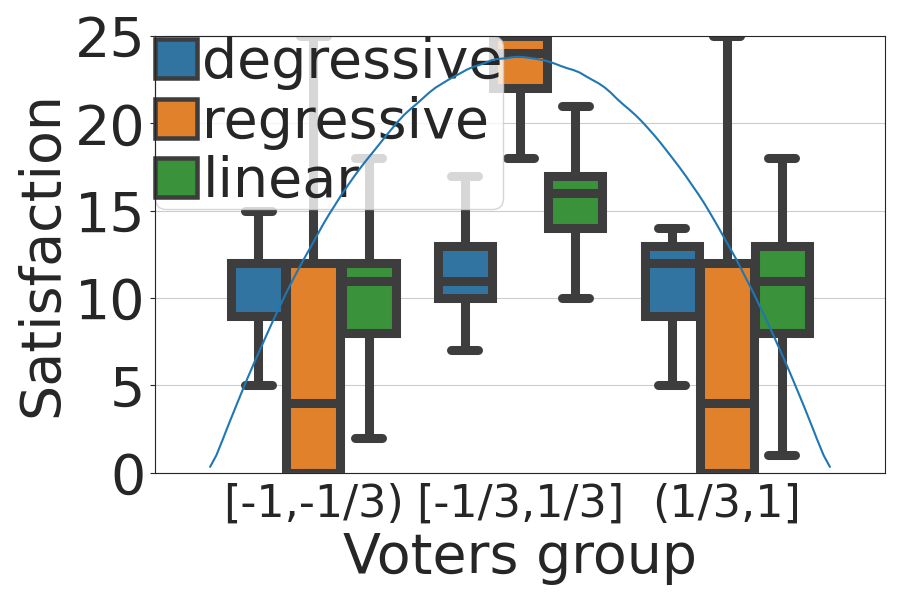}
  (a) \texttt{Beta}$(2,2)$
\endminipage
\minipage{0.245\textwidth}
  \includegraphics[width=\linewidth]{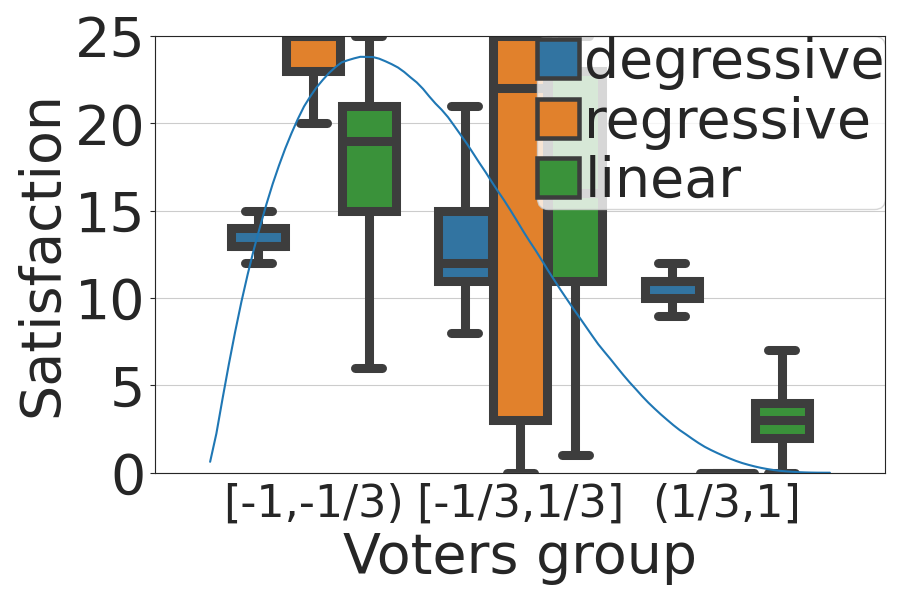}
  (b) \texttt{Beta}$(2,4)$
\endminipage
\minipage{0.245\textwidth}
  \includegraphics[width=\linewidth]{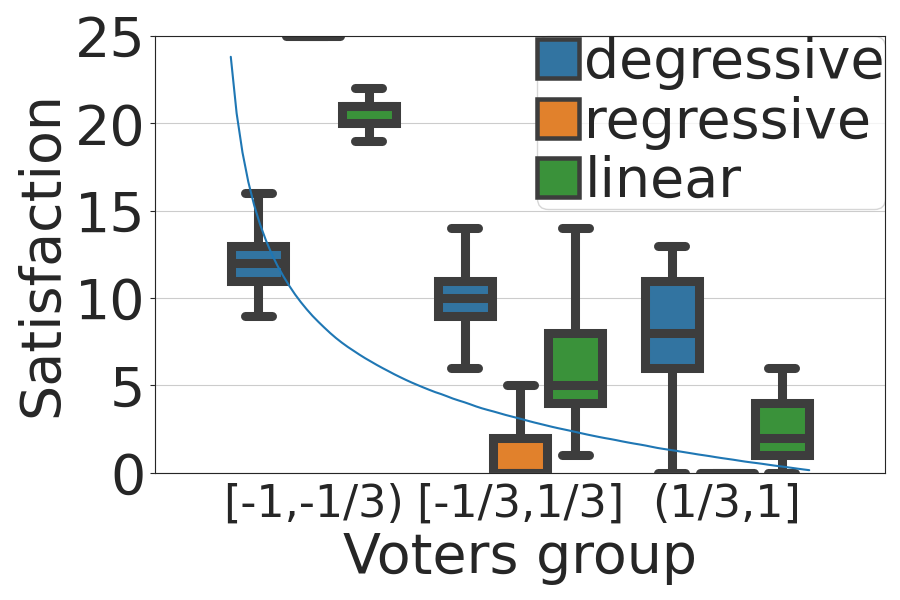}
  (c) \texttt{Beta}$(\nicefrac{1}{2},2)$
\endminipage
\minipage{0.245\textwidth}
  \includegraphics[width=\linewidth]{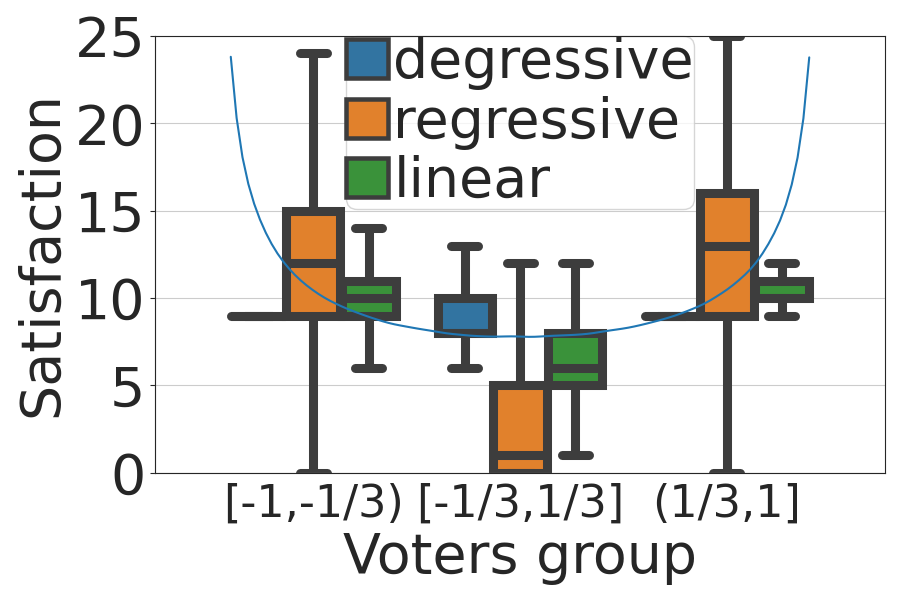}
  (d) \texttt{Beta}$(\nicefrac{1}{2}, \nicefrac{1}{2})$
\endminipage

\end{center}
\caption{Box plots with the distribution of voters' satisfaction for different society models (beta  distributions). Acceptance radius $\xi = 0.4$. In each plot the blue line depicts the density of the distribution from which we sampled voters and candidates.}\label{fig:basic_satisfaction_0.4_appendix}
\end{figure*}

\begin{figure*}[!t]
\begin{center}

\minipage{0.245\textwidth}
  \includegraphics[width=\linewidth]{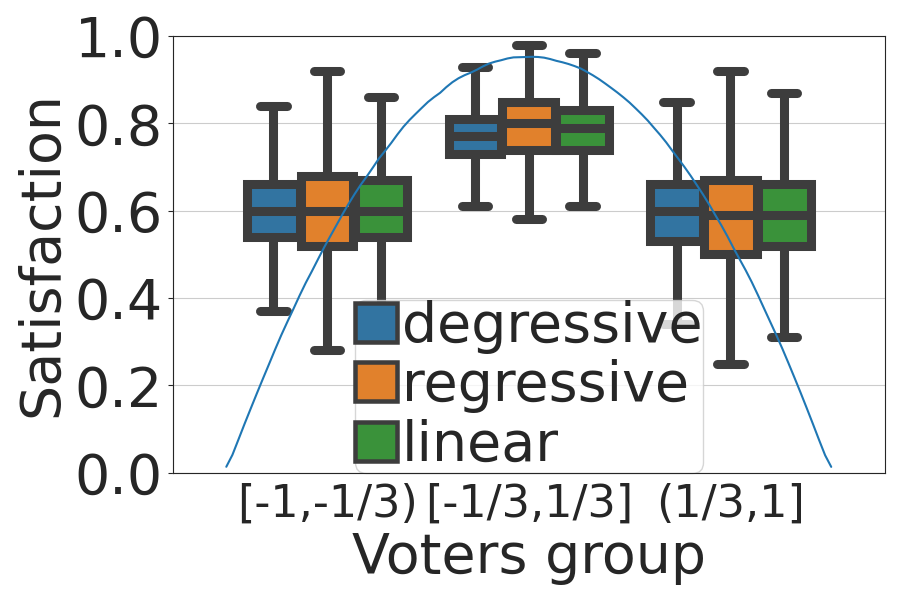}
  (a) \texttt{Beta}$(2,2)$
\endminipage
\minipage{0.245\textwidth}
  \includegraphics[width=\linewidth]{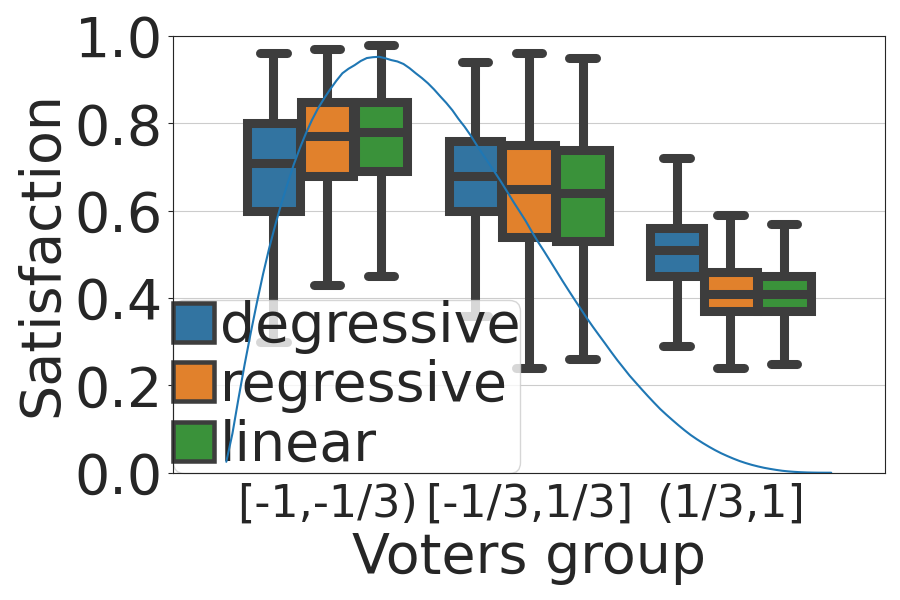}
  (b) \texttt{Beta}$(2,4)$
\endminipage
\minipage{0.245\textwidth}
  \includegraphics[width=\linewidth]{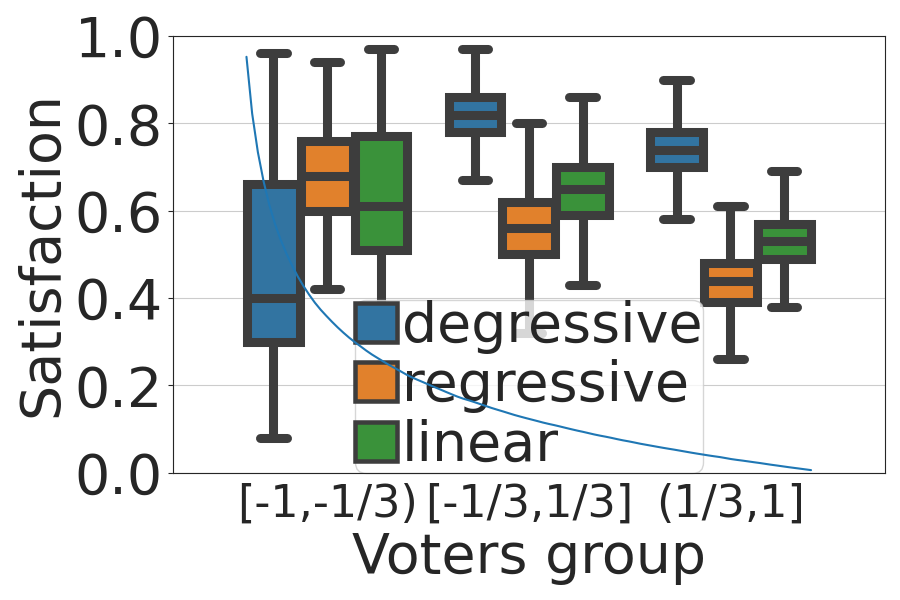}
  (c) \texttt{Beta}$(\nicefrac{1}{2},2)$
\endminipage
\minipage{0.245\textwidth}
  \includegraphics[width=\linewidth]{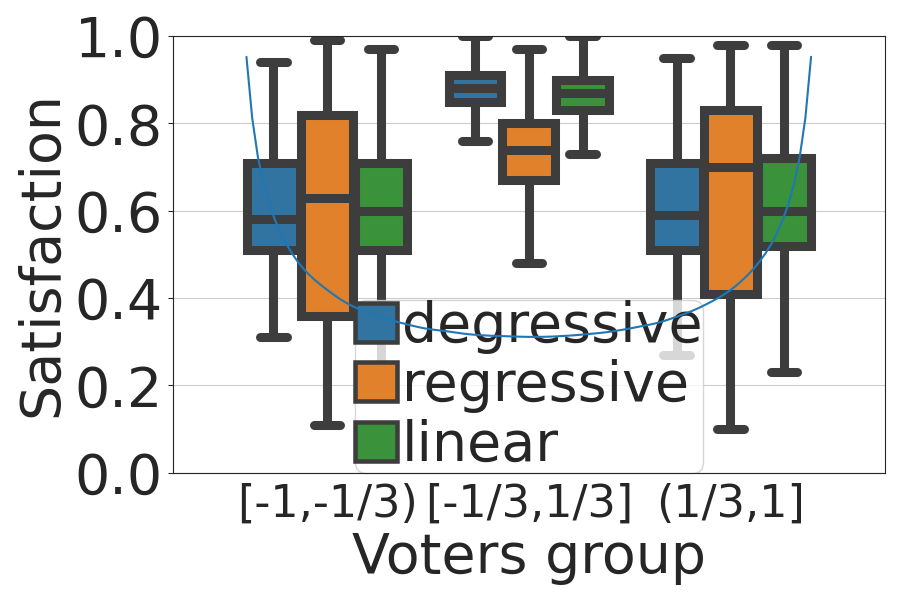}
  (d) \texttt{Beta}$(\nicefrac{1}{2}, \nicefrac{1}{2})$
\endminipage

\end{center}
\caption{Box plots with the distribution of voters' satisfaction in the Voting Committee Model for different society models (beta  distributions). Acceptance radius $\xi = 0.4$. In each plot the blue line depicts the density of the distribution from which we sampled voters and candidates.}\label{fig:issues_satisfaction_0.4_appendix}
\end{figure*}

\begin{figure*}[!t]
\begin{center}
\minipage{0.245\textwidth}
  \includegraphics[width=\linewidth]{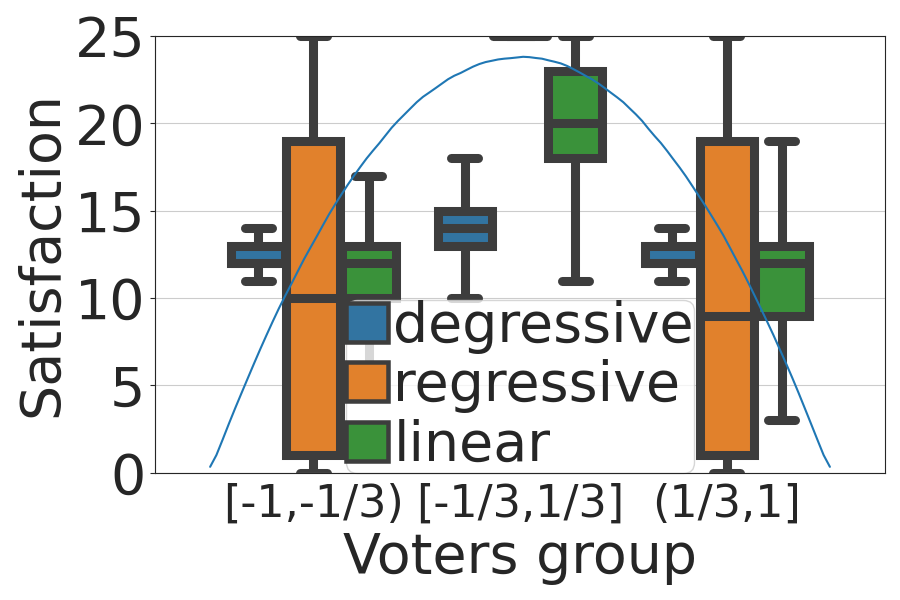}
  (a) \texttt{Beta}$(2,2)$
\endminipage
\minipage{0.245\textwidth}
  \includegraphics[width=\linewidth]{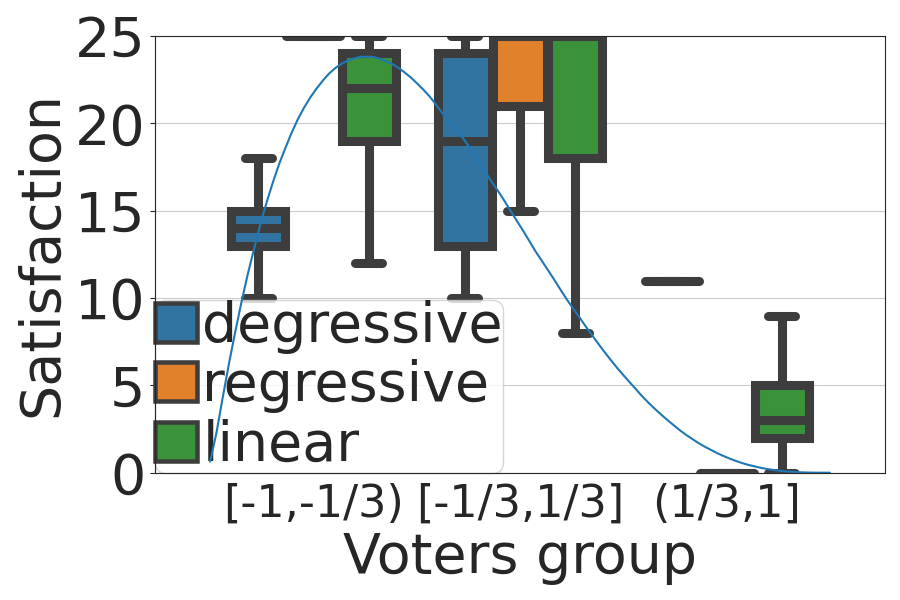}
  (b) \texttt{Beta}$(2,4)$
\endminipage
\minipage{0.245\textwidth}
  \includegraphics[width=\linewidth]{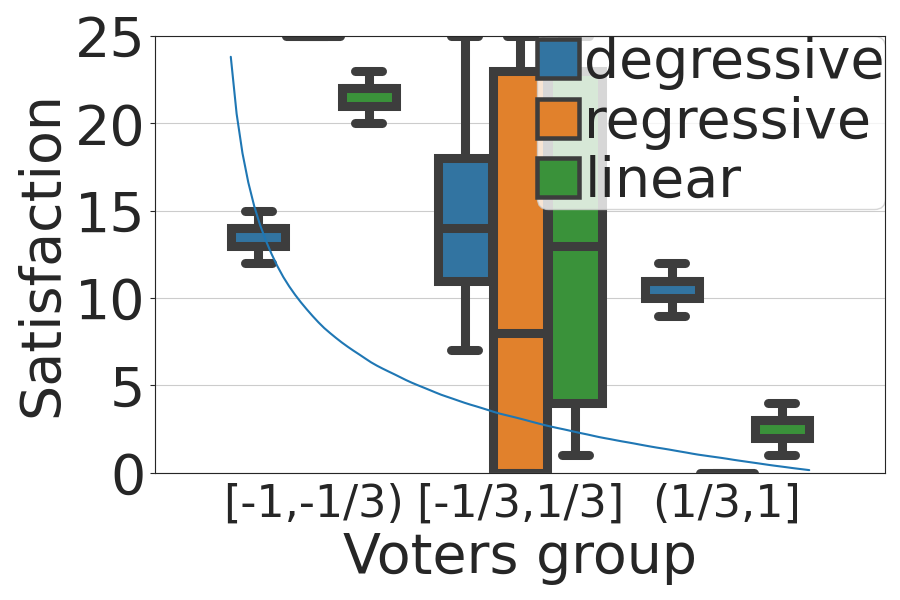}
  (c) \texttt{Beta}$(\nicefrac{1}{2},2)$
\endminipage
\minipage{0.245\textwidth}
  \includegraphics[width=\linewidth]{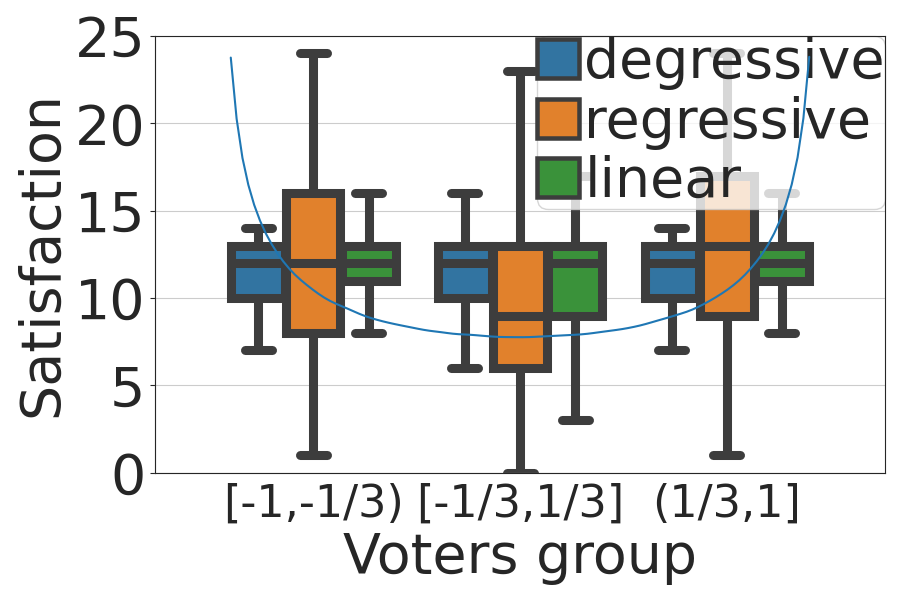}
  (d) \texttt{Beta}$(\nicefrac{1}{2}, \nicefrac{1}{2})$
\endminipage

\end{center}
\caption{Box plots with the distribution of voters' satisfaction for different society models (beta  distributions). Acceptance radius $\xi = 0.5$. In each plot the blue line depicts the density of the distribution from which we sampled voters and candidates.}\label{fig:basic_satisfaction_0.5_appendix}
\end{figure*}

\begin{figure*}[!t]
\begin{center}

\minipage{0.245\textwidth}
  \includegraphics[width=\linewidth]{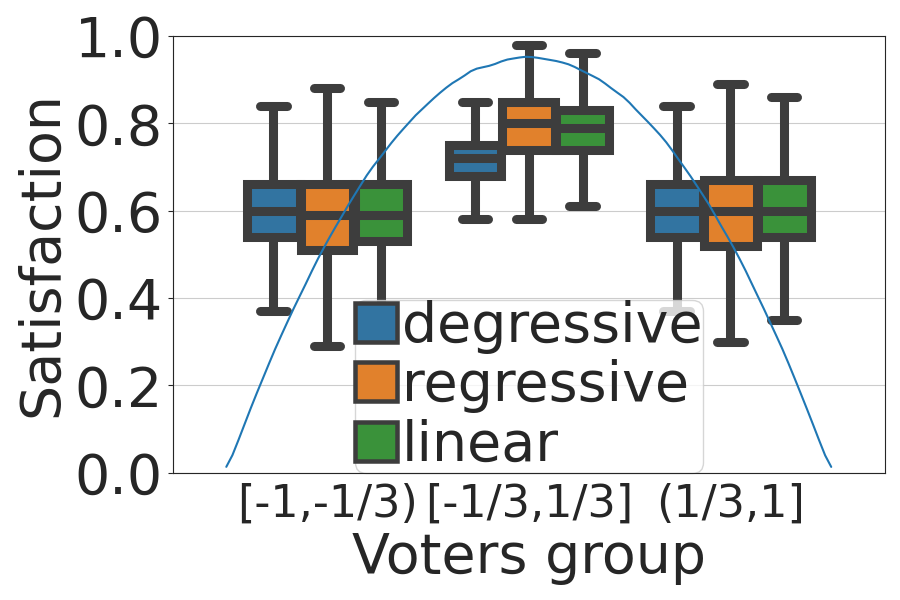}
  (a) \texttt{Beta}$(2,2)$
\endminipage
\minipage{0.245\textwidth}
  \includegraphics[width=\linewidth]{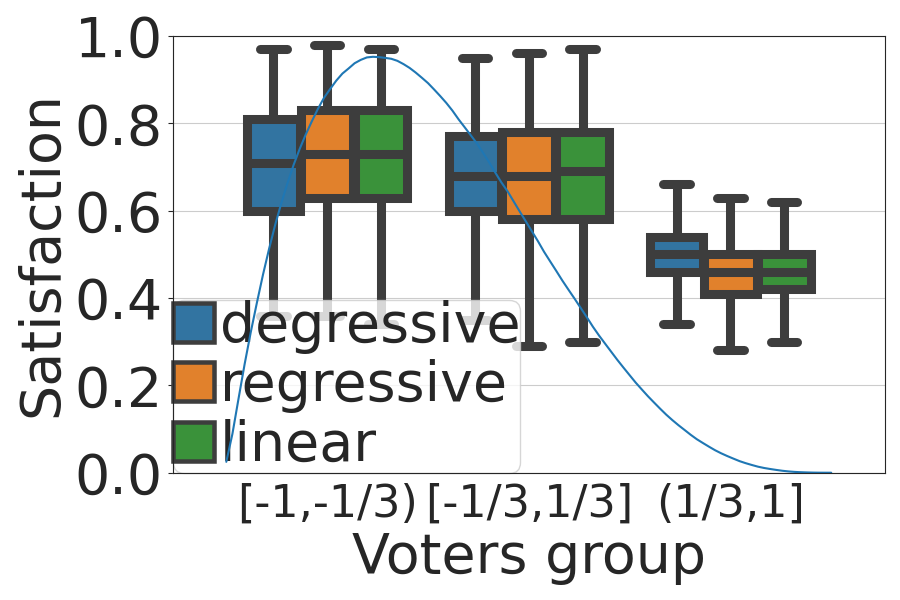}
  (b) \texttt{Beta}$(2,4)$
\endminipage
\minipage{0.245\textwidth}
  \includegraphics[width=\linewidth]{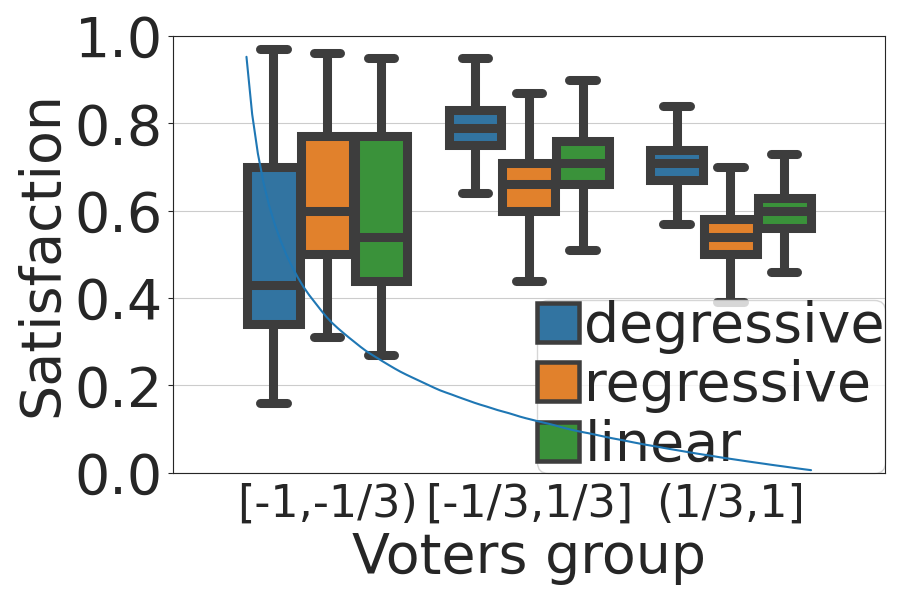}
  (c) \texttt{Beta}$(\nicefrac{1}{2},2)$
\endminipage
\minipage{0.245\textwidth}
  \includegraphics[width=\linewidth]{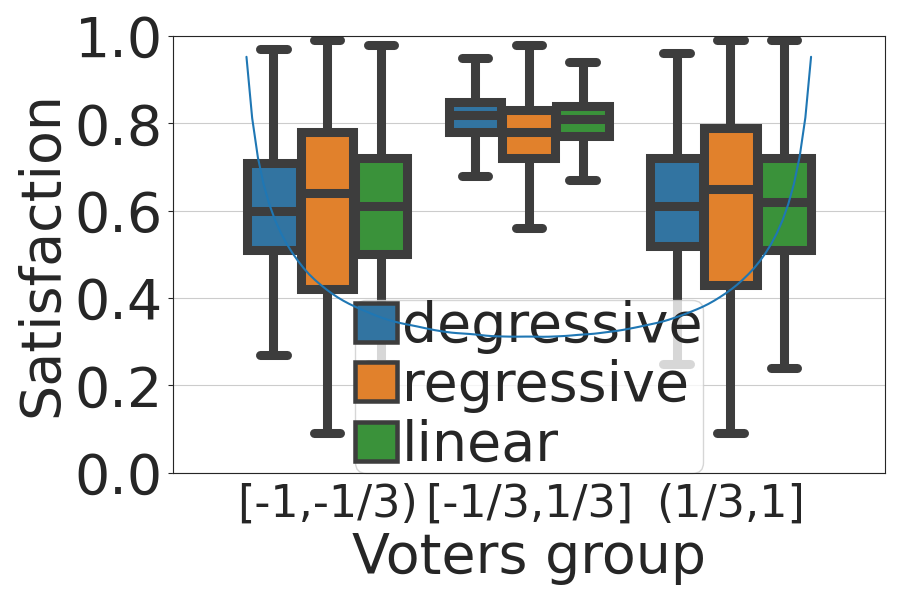}
  (d) \texttt{Beta}$(\nicefrac{1}{2}, \nicefrac{1}{2})$
\endminipage

\end{center}
\caption{Box plots with the distribution of voters' satisfaction in the Voting Committee Model for different society models (beta  distributions). Acceptance radius $\xi = 0.5$. In each plot the blue line depicts the density of the distribution from which we sampled voters and candidates.}\label{fig:issues_satisfaction_0.5_appendix}
\end{figure*}
\clearpage

%
%
%
%







\end{document}